\documentclass[pre,letterpaper,twocolumn,final,superscriptaddress,floatfix]{revtex4}

\usepackage[utf8]{inputenc}
\usepackage[T1]{fontenc}
\usepackage{lmodern}
\usepackage{amssymb, amsfonts, amsthm, amsmath,dsfont}
\usepackage[english]{babel}
\usepackage[pdfborder={0 0 0}]{hyperref}
\usepackage{enumitem}
\usepackage{tabularx}

\usepackage{tikz}
\usepackage{subfigure}

\theoremstyle{plain}
\newtheorem{theorem}{Theorem}[section]

\newtheorem{lemma}[theorem]{Lemma}

\theoremstyle{definition}

\theoremstyle{remark}
\newtheorem{remark}[theorem]{Remark}

\renewcommand{\P}{\mathbb{P}}

\newcommand{\dcens}{d_{\text{cens}}}

\newcommand{\N}{\mathbb{N}}

\newcommand{\ind}[1]{\mathbf{1}_{\left\{#1\right\}}}

\newcommand{\floor}[1]{{\left\lfloor #1 \right\rfloor}}

\newcommand{\refn}[1]{Eq.~(\ref{#1})}
\definecolor{blanchedalmond}{rgb}{1.0, 0.92, 0.8}
\usepackage{mdframed}
\usepackage{ragged2e}

\renewcommand{\bar}[1]{\overline{#1}}

\newcommand{\perf}{\text{(perf)}}
\newcommand{\imp}{\text{(imper)}}

\usepackage{xcolor}

\newcommand{\titlename}[0]{Group testing as a strategy for the epidemiologic monitoring of COVID-19}

\begin{document}

\title{\titlename}
\author{Vincent Brault} 
\affiliation{Université Grenoble Alpes, CNRS, Grenoble INP, LJK, 38000 Grenoble, France,}
\affiliation{Members of the GROUPOOL collective \& Participants in the MODCOV19 initiative.}
\author{Bastien Mallein} 
\affiliation{Université Sorbonne Paris Nord, LAGA, UMR 7539, F-93430, Villetaneuse, France.}
\affiliation{Members of the GROUPOOL collective \& Participants in the MODCOV19 initiative.}
\author{Jean-François Rupprecht} 
\email{vincent.brault@univ-grenoble-alpes.fr; mallein@math.univ-paris13.fr; rupprecht@cpt.univ-mrs.fr}
\affiliation{Aix Marseille Univ, CNRS, Centre de Physique Théorique, Turing Center for Living Systems, Marseille, France.}
\affiliation{Members of the GROUPOOL collective \& Participants in the MODCOV19 initiative.}

\date{\today}

\begin{abstract}
Sample pooling consists in combining samples from multiple individuals into a single pool that is then tested using a unique test-kit. A positive test means that at least one individual within the pool is infected. Here, we propose an analysis and applications of sample pooling to the epidemiologic monitoring of COVID-19. We first introduce a model of the RT-qPCR process used to test for the presence of virus in a sample and construct a statistical model for the viral load in a typical infected individual inspired by the clinical data from Jones et.\@ al.\@ (2020). We then propose a method for the measure of the prevalence in a population, based on group testing, taking into account the increased number of false negatives associated with this method. Finally, we present an application of sample pooling for the prevention of epidemic outbreak in closed connected communities (e.g. nursing homes).
\end{abstract}

\maketitle

Regularly monitoring the \textit{prevalence} of a disease , i.e. the proportion of infected individuals within the general population at a given time, is a key element to prevent the onset of an epidemic wave, to estimate the effect of social distancing policies and to anticipate a potential increase in the demand for hospitalization in intensive care units~\cite{Salje2020}.

In the context of the COVID-19 epidemics, contagious individuals are generally assumed to bear a viral load of SARS-CoV-2 in their respiratory tract~\cite{Wolfel2020,WorldHealthOrganization2020}. Such viral load can be detected and quantified within swab samples using a technique called \emph{reverse transcription quantitative polymerase chain reaction} (RT-qPCR)~\cite{Corman2020}. With tests performed in priority on symptomatic patients, the proportion of positive tests (which corresponds to an \textit{apparent} prevalence~\cite{Speybroeck2013}) is larger than the prevalence among the whole population, which we call  \textit{overall} prevalence. In principle, the overall prevalence could be deduced from the apparent prevalence based on inferred  model estimates for the proportion of tested individuals among the infected population. However, such reliable estimation is challenging given (a) the current large uncertainty regarding the proportion of asymptomatic carriers (estimated to be in $20-50\%$ range~\cite{Mizumoto2020,Bi2020,Bai2020,Lavezzo2020}) and (b) the variable delay between the contamination and first symptoms, which varies from $1$ to $5$ days ~\cite{Ferretti2020,Sethuraman2020}. 

Testing a large portion of the population at random would allow for a direct measure of the overall prevalence, including the proportion of asymptomatic individuals. Unfortunately, it appears that the production of reactants used in RT-qPCR diagnostic would not meet a demand in regular large-scale population testing~\cite{GG,Pouwels2020}.

In such context of a shortage in reactants and/or of RT-qPCR machines, group testing has received renewed interest. The principle of group testing consists in combining samples from multiple individuals into a single pool that is then tested using a single test-kit. The pool sample is considered to be positive if and only if at least one individual in the group is contaminated. The idea of group testing is not new, with a long history that dates back to 1943~\cite{Dorfman1943} in the context of syphilis detection, see~\cite{Aldridge2019} for a review. Optimal diagnostic strategies include smart-pooling, whereby pools are organised according to lines and columns on a grid --or hypercube-- with overlaps enabling the identification of positive individuals~\cite{Barillot1991,Thierry-Mieg2006,Furon2018}.

Several teams across the world have developed group testing protocols for SARS-CoV-2 infected individuals using RT-qPCR tests. As early as February 2020, pools of 10 have been used over 2740 patients to detect 2 positive patients over the San Francisco Bay in California~\cite{Hogan2020}. A recent publication from Saarland University, Germany, shows that positive sample with a relatively mild viral load from asymptomatic patients could still be detected within pools of 30~\cite{Lohse2020}. Further works  suggest that RT-qPCR viral detection can been achieved in pools with a number of samples ranging from $5$ to $64$ ~\cite{Yelin2020,Ben-Ami2020,Shental2020,Schmidt2020,Mutesa2020,Torres2020,Cabrera2020,Wacharapluesadee2020,Martin,Khodare2020,Wernike2020,Gan2020,Griesemer2020}.

In parallel, the theoretical literature on group testing for SARS-CoV-2 diagnostic is growing at a fast pace~\cite{Hanel2020,GG,Deckert2020,SinnottArmstrong2020,Verdun2020,Narayanan2020,Bilder2020}. Most of the emphasis has been put on the binary (positive or negative) outcome of tests, with  little or no regard on the viral load quantification~\cite{Corman2020}. Moreover, if the possibility of false negatives is sometimes considered, the increase in the rate of false negatives with dilution of samples due to group testing is rarely taken into account~\cite{Furon2018}. 

In this article, we do not address any diagnostic problems, such as the question of determining optimal strategies to provide individual positive diagnostic to a large population using a minimal number of tests. Rather, we propose to evaluate pooling strategies as a tool for the study of epidemiologic questions. In the pooling strategies we discuss below, no individual will be part of two different pools at the same time, so no information on the infection status of any distinguished individual is obtained. 

Here, we instead focus on (i) the measure of the overall prevalence and (ii) on the early detection of contamination in a closed community.

The rest of the article is constructed as follows. We first present a simple protocol for the measure of the prevalence in the population by the use of group testing; we make in Section~\ref{sec:sc} the assumption of \emph{perfect test}, i.e. that the test used is not subject to any false positive or negative, no matter the size of the pool to which the test is applied. In Section~\ref{sec:modeling}, we provide a short description of the RT-qPCR and propose a statistical model for its study, that underlines its limit of detection at small concentrations. Then in Section~\ref{sec:statistical}, we analyse part of the information recovered from the quantified viral charge in patients from the clinical dataset of \cite{Jones2020}; this gives a sense of the viral load infected patients should carry in the general population, therefore the effect of dilution on the rate of false negatives. Finally, we show in Section~\ref{sec:epidemiology} how, using the statistical model and the measure of the error rate discussed above, one can measure the viral prevalence in the general population, or design a protocol allowing an early detection of an epidemic outbreak in a closed vulnerable community (e.g. schools, retirement homes, detention centers).

\section{Measuring prevalence with perfect tests}
\label{sec:sc}

We investigate in this section the measure of the prevalence of the disease in a population using a group testing strategy, under the assumption of \textit{perfect} tests, i.e. with no risks of false negative (or false positive). Our derivation is similar to~\cite{Thompson1962}.

We assume that we have $n$ tests at our disposal. Given $N \in \N$, we sample $nN$ individuals at random in the general population, and organize $n$ pools of $N$ individuals. Each of these pools is then tested using the perfect tests. For all $i \leq n$, we write $X^{(N)}_i = 1$ if the $i$th test is positive (i.e. if and only if at least one of the $N$ individuals in the $i$th pool is infected), and $X^{(N)}_i = 0$ otherwise. We denote by $p$ the (unknown) proportion of infected individuals in the population, then $(X^{(N)}_i, i \leq n)$ forms an independent and identically distributed (i.i.d.) sequence of Bernoulli random variables with parameter $1-(1-p)^N$.
\begin{lemma}
Writing $\bar{X}^{(N)}_n = \frac{1}{n} \sum_{j=1}^N X^{(N)}_j$, the quantity $1 - (1-\bar{X}^{(N)}_n)^{1/N}$ is a strongly consistent and asymptotically normal estimator of $p$. A confidence interval of asymptotic level $1-\alpha$ is
\begin{align}
  \label{eqn:ci}
  &\mathrm{CI}_{1-\alpha}(p)  =  \bigg[ 1 - (1 - \bar{X}^{(N)}_n)^{1/N}  \bigg.  \nonumber \\ 
  &  \bigg. \pm \frac{q_\alpha (1 - \bar{X}^{(N)}_n)^{1/N - 1}{\sqrt{\bar{X}_n^{(N)} (1 - \bar{X}^{(N)}_n)}}}{\sqrt{n} N}  \bigg],
\end{align}
where $q_\alpha$ is the quantile of order $1-\alpha/2$ of the standard Gaussian random variable. \label{lemma1}
\end{lemma}

\begin{proof}
Note that $(X^{(N)}_j, j \leq n)$ is a standard Bernoulli model, hence $\bar{X}^{(N)}_n$ is a consistent and asymptotically normal estimator of $f(p) = 1 - (1-p)^N$. Hence, using that $f^{-1}$ is $\mathcal{C}^1$ and Slutsky's lemma, we deduce all the above properties of the estimator $f^{-1}(\bar{X}_n^{(N)}) $ of $p$.
\end{proof}

\begin{remark}
\label{rem:widthEquation}
As $\lim_{n \to \infty} 1 - (1-\bar{X}^{(N)}_n)^{1/N} = p$ almost surely, for any $N \in \N$ the width of the confidence interval defined in Lemma~\ref{lemma1} satisfies
\begin{align}
  \label{eqn:widthEquivalent}
  &\frac{2 q_\alpha (1 - \bar{X}^{(N)}_n)^{1/N - 1}{\sqrt{\bar{X}_n^{(N)} (1 - \bar{X}^{(N)}_n)}}}{\sqrt{n} N}  \nonumber\\
  \underset{n \to \infty}{\sim}&\frac{2 q_\alpha}{\sqrt{n}} \frac{(1-p)}{N} \sqrt{\frac{1- (1-p)^N}{(1-p)^N}} \quad \text{a.s.}
\end{align}
In other words, the precision of the measure of prevalence decays as $n^{-1/2}$, with a prefactor that depend on the prevalence $p$ and the number $N$ of individual per pool. There exists an optimal choice of $N$ that minimizes the value of this prefactor, largely improving the precision of the measure. 
\end{remark}
A classical computation (cf.\@ e.g.\@~\cite{GG}) shows that the prefactor in~\refn{eqn:widthEquivalent} is minimal when the number of mixed samples per pool is equal to:
\begin{equation}
  \label{eqn:nopt}
  N^{\perf}_\text{opt} = -\frac{c_\star}{\log(1-p)} \iff (1-p)^{N^{\perf}_\text{opt}}  \approx 0.20,
\end{equation}
where $c_\star = 2 + W(-2e^{-2})\approx 1.59$ and $W$ is the Lambert $W$ function. Specifically, the size of the pools is optimal when approximately 80\% of the tests made on the groups turn positive, in sharp contrast with the Dorfman criterium \cite{Dorfman1943}.

\begin{figure}[t!]
\centering
\includegraphics[width=8.5cm]{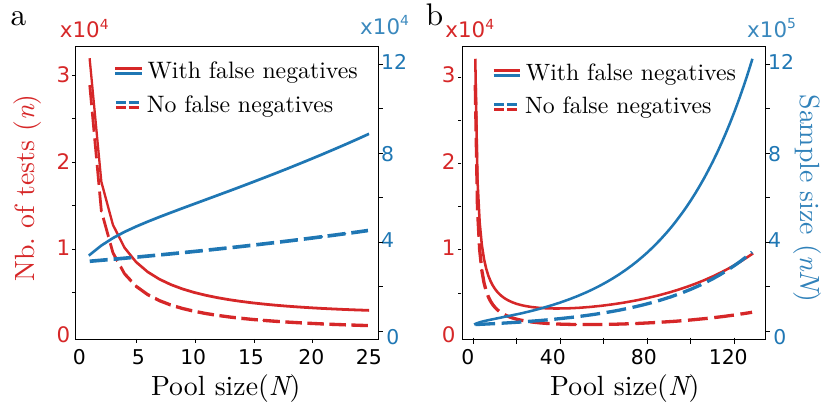}
\caption{(a,b) Total number of tests (red) and total number of sampled individuals (blue) in order to estimate a prevalence of $p = 3\%$ with a $\pm 0.2\%$ precision with $95\%$ confidence interval as a function of the pool size $N$ for the perfect case (dashed lines) considered in Sec.~\ref{sec:sc} with no false negative, and the more realistic case (solid lines) considered in Sec.~\ref{sec:epidemiology} (with false negatives; parameters are defined in Table~\ref{tab:monitoring}). In (a) $N$ ranges from $0$ to $25$; in (b) $N$ ranges from $0$ to $128$; as visible in (b), the valley around the optimal pool size $N^{\perf}_{\mathrm{opt}} \approx 50$ is large: near optimal savings in tests are achieved even for moderately large pool sizes that require smaller number of individuals to sample.}
\label{fig_prevalence}
\end{figure}

If we measure the prevalence of the population using group testing, choosing $N = N^{\perf}_\text{opt}$ for the size of the groups, then measuring with a given precision the prevalence will require significantly less tests than if we were to use one test per sampled individual (i.e. if $N=1$). On the other hand, using this group testing method increases the total number of individuals needed to be sampled, which also has a cost to be considered. However, one can observe that the bottom of the valley of the (red) functions plotted in Fig.~\ref{fig_prevalence}, that represent the number of tests needed as a function of the size of the pool, is rather wide and flat. There is therefore a large variety of quasi-optimal pool sizes that can be chosen with minimal diminution of the precision in the measure of the prevalence.

Taking the exemple of a prevalence at $p = 3\%$, we expect the pool size minimizing the number of tests needed to read $N^{\perf}_\text{opt}=50$, which divides by $20$ the number of tests needed (see Fig. \ref{fig_prevalence}). But to reach similar level of precision than in single testing, the total number of individuals that need to be sampled is more than doubled. Choosing instead a pool size of $N=20$ requires almost the same number of tests, yet at a cost of a $30\%$ in the total number $Nn$ of sampled individuals (c.f. Fig.~\ref{fig_prevalence}). The same observation holds for different values of the prevalence, see SI. Fig.~S1.

\section{Statistical modelling of the PCR}
\label{sec:modeling}

The RT-qPCR technique has been extensively used to estimate the concentration of viral material in samples~\cite{Forootan2017}.
We briefly sketch the main steps of an RT-qPCR diagnostic protocol in Box 1. The qPCR typically returns a $C_t$ value, which corresponds to $-\log_2$ of the initial number of DNA copies in the sample, up to an additive constant and measure error. It is measured as an estimated number of cycles needed for the intensity of the fluorescence of the sample to reach a target value (see Fig.~\ref{fig:PCRdescription}).

Combined measures of two viral RNA strands with a control of a human RNA strand are recommended in order to detect defective sampling that could induce false negatives, but also to improve precision of the measure as well as to normalize the number of virus copies by the quantity of human DNA~\cite{Corman2020}. Such combined measure can also improve precision of the measure as well as to normalize the number of virus copies by the quantity of human. We do not intend to include such features in our model.  Furthermore, we do not model here the possible errors at the reverse transcription stage, which could lead to some biased measure of the viral load distribution

PCR tests are prone to amplify non-specific DNA sequences~\cite{Forootan2017,Ruiz-Villalba2017} that can trigger an onset of fluorescence in a samples with no viral SARS-CoV-2 load. These events, called \textit{artefacts}~\cite{Ruiz-Villalba2017}, will typically occur beyond a relatively large critical number of cycles, thus imposing the following condition on the diagnosis: a reliable positive result can only be made if the $C_t$ value is lower than a critical value, denoted $\dcens$. Here, the onset of fluorescence in which the virus is absent will be modelled as if triggered by a vanishingly small artificial  concentration, denoted $\epsilon_1$.


\subsection{Statistical model for the cycle threshold value for a fixed viral concentration}

We propose to model the number of cycles threshold value $C_t$ as a random variable, denoted by $Y$, that depends on the viral load $c$ in the measured sample as
\begin{equation}
Y = - \log_2\left( c + \epsilon_1 \right) + \epsilon_2,
\label{eq:modelY}
\end{equation} 
where $\epsilon_1$ is the law of the artefact, modelled as a log-normal distribution with parameters $(\nu,\tau^2)$;  $\epsilon_2$ is an intrinsic measurement error on the threshold value $C_t$ measurement, modelled as a centered Gaussian random variable with variance $\rho^2$.

As mentioned above, tests are considered to be reliably positive when $Y \leq \dcens$. To avoid false positives, the threshold $\dcens$ (with cens for censoring) is chosen such that $\P(\epsilon_1> 2^{-\dcens})\ll 1$. Thus, using that as long as $a$ and $b$ are of different orders of magnitude, we have $\log(a+b) \approx \log(\max(a,b))$, we deduce that
\begin{equation}
Y \approx \min(-\log_2(c),\dcens) + \epsilon_2,
\label{eq:assumption}
\end{equation}
which obeys the law of a Gaussian random variable with variance $\rho^2$ and mean $-\log_2(c)$, censored at $\dcens$.

In the idealized no artefact limit ($\epsilon_1 \rightarrow 0$), the PCR threshold intensity of a negative patient ($c=0$) would never be reached ($Y \rightarrow \infty$ as well as $\dcens = \infty$).

\begin{figure}[t!]
\noindent\fcolorbox{black}{blanchedalmond}{
\begin{minipage}{8 cm}
\begin{center}
{\color{blue}{\textbf{Box 1: A brief description of RT-qPCR tests}} \vskip-0.3cm}
\end{center} 
\justify
\noindent    
We very briefly review some of the steps implemented during an RT-qPCR diagnostic procedure~\cite{Corman2020}:
\begin{enumerate}[leftmargin=*]
  \item The sample is treated so that a target RNA sequence (characteristic of the virus) is transcribed into DNA (reverse transcription);
  \item The sample is placed in a PCR machine, which can measure the concentration of DNA of interest in the sample by making it fluorescent;
  \item A reactive is added which approximatively doubles the number of DNA of interest at every cycle, driven by temperature changes;
  \item The time series of the concentration in DNA over time is recorded;  on a linear regression of of the logarithm of the fluorescent signal over time, one deduces an estimate of the viral concentration in the sample from the linear regression value at the origin.
\end{enumerate}
\end{minipage}}
\end{figure}

\subsection{Model of the cycle threshold values for pooled samples}
We now consider what happens when constructing a pooled sample of $N$ samples. For each $i \leq N$, we write $Z_i=1$ if the sample $i$ contains a viral RNA load with concentration $C_i > 0$, and $Z_i=C_i=0$ otherwise. In the rest of the paper, we assume that, in a combined sample created from a homogeneous mixing of the individual samples, the viral concentration reads:
\begin{equation}
C^{(N)} = \frac{1}{N}\sum_{j=1}^N Z_j C_j.
\end{equation}  \label{eqn:melange}
This assumption relies on the fact that contaminated individuals should have a sufficiently high number of viral copies per sample, so that taking a portion $1/N$ of a virus bearing sample brings a fraction $1/N$ of its viral charge. The result of the RT-qPCR measure of a grouped test with $N$ individuals is then given by Eq. \eqref{eq:modelY}, with $c=C^{(N)}$, hence reads
\begin{equation}
  \label{eqn:groupModel}
 Y^{(N)} = \min\left( \log_2 N - \log_2\left( \sum_{j=1}^N Z_j C_j\right), \dcens \right) + \epsilon_2.
\end{equation}
where $(Z_i, i \leq N)$ are i.i.d. Bernoulli random variables whose parameter is the prevalence of the disease in the population;  $(C_i , i \leq N)$ are i.i.d. random variables corresponding to the law of the viral concentration within samples taken from a typical infected individual in the overall population.

Our model Eq. \eqref{eqn:groupModel} is consistent with the experimental result of \cite{Yelin2020} as well as \cite{Gan2020}, whereby linear relations are found between the logarithm of the pool size and the measured $C_t$ that are sufficiently distant from the identified detection threshold.

\begin{remark}
If it were possible to combine samples without dilution (e.g. following the protocol of \cite{Griesemer2020}, whereby the exact same volume of each sample is added to the buffer solution as if the sample were being tested individually), \refn{eqn:groupModel} would then be replaced by
\begin{equation}
    \label{eqn:gmIdeal}
    \tilde{Y}^{(N)}
    = \min\left(-\log_2 \left(\sum_{j=1}^N Z_j C_j\right), \dcens \right) + \epsilon_2,
\end{equation}
in which case, theoretically, pool testing would never loose precision when the pool size increases. That case was treated in Section \ref{sec:sc}. Hoewever, if the dilution effect occurs for pool sizes exceeding a thresold size $K$,  Eq. \eqref{eqn:groupModel} would be replaced by
\begin{multline}
      \bar{Y}^{(N,K)} =  \label{eqn:gmLessIdeal}\\
       \min\left( \log_2(\tfrac{N}{K})_+ - \log_2\left( \sum_{j=1}^N Z_j C_j \right),\dcens \right)  + \epsilon_2;
\end{multline}
where $\log_2(N/K)_+ = 0$ if $N < K$ and $\log_2(N/K))$ otherwise; the analysis would then be similar to what is presented in the rest of the paper, yet with a lower false negative rate. 
\end{remark}

In order to determine the statistics of the measured cycle $Y^{(N)}$ in a group test of $N$ individuals, we need a distribution for the value of $C_j$, the viral distribution of infected individuals in the population; this is the objective of the next section.

\begin{figure}[t!]
\centering
\includegraphics[width=8.5cm]{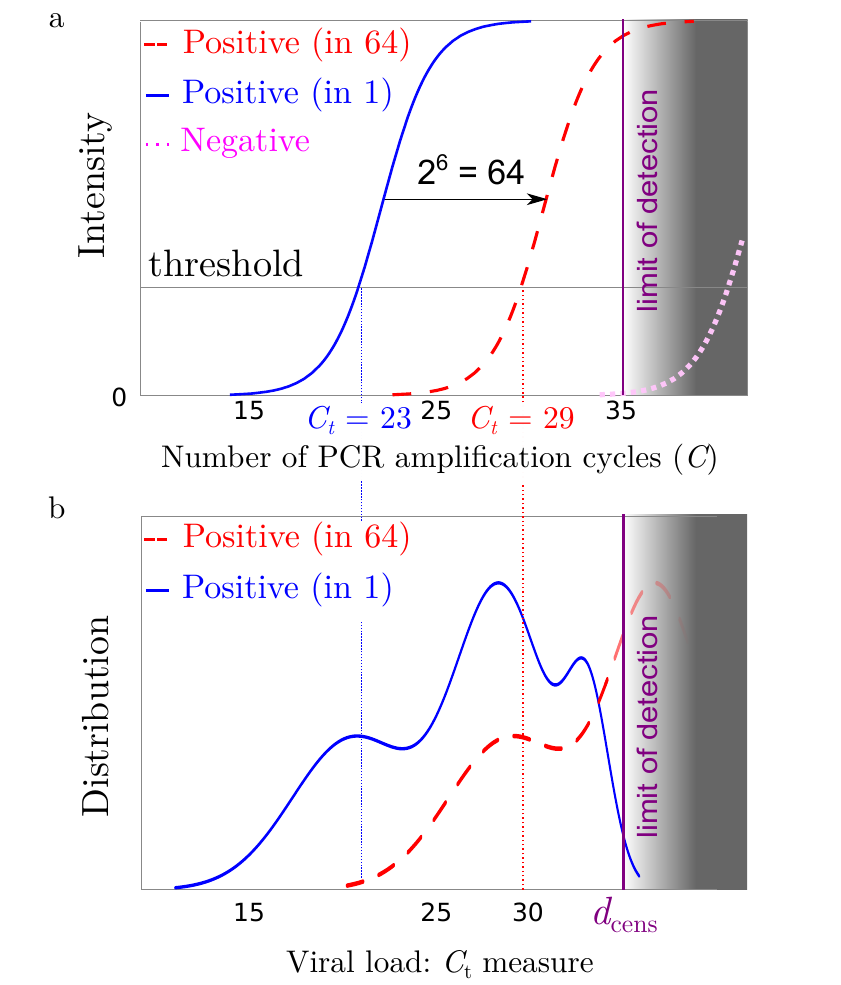}
\caption{(a) Sketch of an RT-qPCR fluorescence intensity signal for a positive patient without pooling (solid red line) a single positive patient in a pool of 64 patients (dashed red curve) and for a negative sample representing the response of an artefact (dotted magenta curve); as pooling dilutes the initial concentration, the pooled response (dashed red curve) is expected to be close to the translation $x \rightarrow x + 6$ from that of a single patient  (solid red line).
(b)~Sketch of the distribution of threshold values for qPCR, either for individual testing (solid blue line) or in pools $64$ (dashed red curve); part of the distribution crosses the limit of detection of the test (figured as the grey area) at the detection threshold $\dcens$.}
\label{fig:PCRdescription}
\end{figure}

\section{Statistical analysis of the viral load measured in positive samples}
\label{sec:statistical}

In~\cite{Jones}, the authors propose an analysis of SARS-CoV-2 viral load by patient age. Their analysis is based on the clinically measured viral load of a series of 3,712 infected patients. In particular, in Fig.~1, an histogram representing the frequency distribution of the viral load estimated using RT-qPCR testing. 
We present in this section a method to analyse these data as a mixture of Gaussian random variables. As RT-qPCR does not allow the measure of viral load below a certain value, we take this fact into account in our statistical modelling by considering the Gaussian in the mixture to be censored after a given threshold $\dcens$.

Unfortunately, as the clinical dataset of~\cite{Jones} is not available,  we created a simulated set of measures which reproduces the distribution described by the histogram in Fig.~1~\cite{Jones}. As the precise distribution of data points within each class of the histogram of Fig.~1~\cite{Jones} is unknown, we assumed that points were distributed uniformly in their class. We verified the robustness of our estimator for several realizations for the simulated data\-sets, which lead to consistent values for our model parameters (see SI Sec. 2. B).

Furthermore, as we expect the measure error $\epsilon_2$ of the qPCR to be small with respect to the width of the histogram classes, we set $\rho = 0$ in the rest of the section.

The histogram of our simulated data is plotted in Fig.~\ref{fig:distribution_viral_load}. It is similar to the histogram~\cite[Fig.~1]{Jones} but with the $x$-axis graduated as $-\log_2$ of the viral concentration (to obtain an estimation of the attended $C_t$ value). We observe the presence of a wide and rather low first bump, centered around the value of 20. Additionally, two taller but less wide bell shapes seem to be present, centered around the values 29 and 34 respectively. After the value of about 35.5, there is a significant drop in the number of registered values, very likely due to a sharp detection limit.

The presence of three well-marked distinct bell shapes in the histogram suggests to model the density of $-\log_2$ of the viral load of contaminated patients as a mixture of Gaussian distributions. Heuristically, the patient population can be divided into several groups, which we call clusters, such that patients in cluster $i$ have a $C_t$ value distributed according to a normal distribution with parameters $(\mu_i,\sigma_i)$. 

In a first paragraph A., we use a classical statistical tool to estimate the number of clusters and their parameters (frequencies, means, variances). 
However, this first Gaussian modelling does not take into account the sudden drop observed for values around $C_t \sim 35.5$, hence the estimated density does not fit well to such drop. This seems to be due to the sudden loss of sensitivity of qPCR measure for low viral load. 

To better represent the variables, we thus introduce a (partially) censored Gaussian model in Section~\ref{Sec:censored}, which, to our knowledge, was not previously proposed.

Finally in Section~\ref{Sec:mixing:censure}, we construct a model consisting of a mixing of censored Gaussian variables, that we apply to the data of~\cite{Jones}. This model fits better the data with stable estimation of the parameters and provides further validation of the previous three-cluster analysis of the tested population.


\begin{figure}[t!]
\centering
\includegraphics[width=8.5cm]{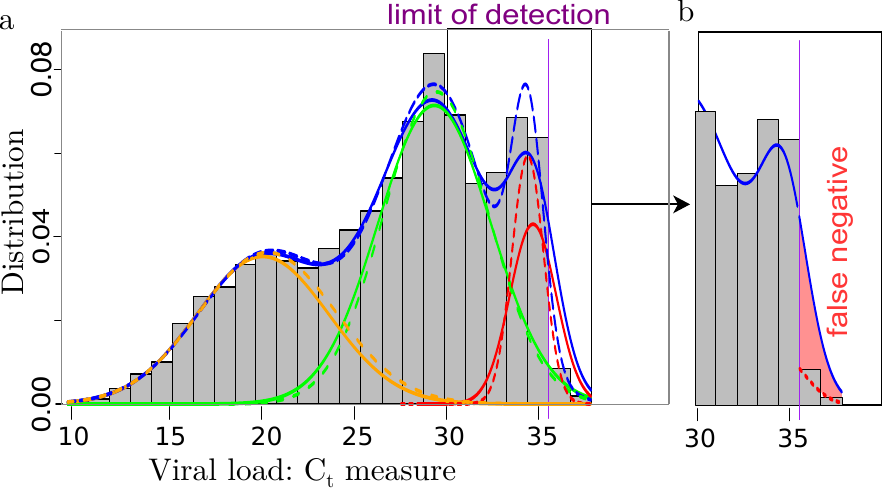}
\caption{(a) Representation of the density for the classical mixing Gaussian model (dashed lines) and the partially censored model (solid lines) each composed as a sum of $3$ components  for the Gaussian model (orange/green/red dashed lines) and the partially censored model (orange/green/red solid lines); (purple vertical line) location of the threshold $\dcens\approx35.6$. (b) Focus on the false negative region, with the estimated false negative probability in the partially censored model (solid line) due to the defect of detection above the threshold $\dcens$ (red color filled area).}
\label{fig:distribution_viral_load}
\label{Fig:Histo:MM_Censure}
\end{figure}

\subsection{Mixture model}\label{sec:Naive:MM}

The shape of the histogram Fig.~\ref{fig:distribution_viral_load} suggests that the law of the viral load is distributed according to a mixture of three or more Gaussian distributions. The typical decomposition into three clusters is represented in dashed lines in Fig.~\ref{fig:distribution_viral_load}(a). We also plotted in SI. Fig.~S4 
the histogram with an example of the estimated density with 3 and 4 clusters.

We observe that the cluster associated to the lowest viral load, i.e. the highest $C_t$ value (red curve in Fig. \ref{fig:distribution_viral_load}) has a small variance.  However, as recalled by~\cite{Jones}, there is a loss of sensibility of the measure for very small viral loads, which can explain the drop in the number of detected cases around $\dcens = 35.6$. This motivates the introduction of the censored Gaussian model in the next section.

\subsection{Censored model and partially censored model}\label{Sec:censored}

To model a partial lack of detection of low viral load ($C_t$ higher than a threshold $\dcens$), we introduce the partially censored Gaussian variable as a building block for the representation of the density of the viral load in infected patients.
However, depending on the machine tuning, the $C_t$ value of the sample can shift by an additive constant, which is computed by measuring the $C_t$ value of a standard solution of viral DNA to tare the measure. Then, some tests might allow the detection of lower viral loads than others.

In view of the shape of the histogram Fig.~\ref{fig:distribution_viral_load}, it is reasonable to assume that the $C_t$ value of patients in a cluster follows a Gaussian distribution. To model this partial censorship phenomenon, we assume that if the sampled $C_t$ value is lower than the detection threshold $\dcens$, the measure is always made for samples detected as positive. If the value is higher than $\dcens$, the sample will be detected with probability $q$, and its measure will be registered. Otherwise, it will be discarded as a (false) negative, with probability $1-q$. The parameter $q$ represents the probability of detection of a viral load that falls below the detection threshold of some PCR measures. 

\begin{remark}
\label{rem:pbCensorship}
The assumption that the probability of detection only depends on whether the $C_t$ value is higher than a fixed threshold is of course an important simplification, as one would expect lower viral loads to be more difficult to detect than higher ones. However, the simplicity of this model allows us to study it as a three parameters statistical model, and to construct simple estimators for these parameters. Additionally, it fits rather well the available data, and fitting a more complicate censorship model would require a lot of measures of $C_t$ values close to the detection threshold $\dcens$.
\end{remark}

We call this statistical model a \textit{partially censored Gaussian model}, denoted by $\mathcal{CN}_{\dcens}(\mu,\sigma,q)$, with $\mu$ and $\sigma$ the mean and standard deviation of the Gaussian variable before censorship and $q$ the detection probability above the threshold. If we denote by $X$ the random variable, $f_{\mu,\sigma}$ (resp. $F_{\mu,\sigma}$) the density (resp. the cumulative distribution function) of a Gaussian law $\mathcal{N}(\mu,\sigma)$ then the density of $X$ is defined for every $x\in\mathbb{R}$ by:
\begin{equation}
f_X(x)=\frac{f_{\mu,\sigma}(x)}{q+(1-q)F_{\mu,\sigma}(\dcens)}\times\left\{\begin{array}{l}
1\text{ if }x\leq\dcens,\\
q\text{ otherwise}.\\
\end{array}\right.
\label{eq:density}
\end{equation}

\begin{remark}
In the absence of censorship (i.e. in the limits $q\rightarrow1$ or $\dcens\rightarrow+\infty$), we check that~\refn{eq:density} converges to a Gaussian density distribution.
\end{remark}

To avoid the problem of modelling of the partial censorship described in Remark~\ref{rem:pbCensorship}, a solution that we implement here as a comparison tool, is \textit{to forget} the values after the threshold and to fit a \textit{completely censored} model (i.e. with $q=0$) to the remaining data, that we denote by $\mathcal{CN}_{\dcens}(\mu,\sigma) = \mathcal{CN}_{\dcens}(\mu,\sigma,0)$, with density defined for every $x\in\mathbb{R}$ by
\begin{equation}
f_X(x)=\frac{f_{\mu,\sigma}(x)}{F_{\mu,\sigma}(\dcens)}\mathds{1}_{\{x\leq \dcens\}},
\label{eq:viral_load_distribution}
\end{equation}
where $\mathds{1}_{\{x\leq \dcens\}}$ is the indicator function equal to $1$ if $x\leq \dcens$, and $0$ otherwise.

Due to the presence of the cumulative distribution function of a Gaussian law in the denominator in the normalization constant, it is not possible to obtain analytical forms of the parameter estimators. Nevertheless, we can estimate the parameters using an optimization algorithm like the \textsf{R} function \textsf{nlm} (available in~\cite{rTest}) which implements a Newton-type algorithm. The following Theorem~\ref{lemma:censure} guarantees the quality of the maximal likelihood estimators, hence of the estimations if the maximization procedure is done correctly.

\begin{theorem}
The estimators $\left(\widehat{\mu},\widehat{\sigma},\widehat{q}\right)$ of $\left(\mu,\sigma,q\right)$ obtained by maximisation of the likelihood ratio are strongly consistent and asymptotically normal.
\label{lemma:censure}
\end{theorem}

The properties of the maximum likelihood estimators is a consequence of the fact that the (partially) censored Gaussian model belongs to the family of exponential laws (c.f.~\cite[Chapter 9]{stats} and SI~B. 
). To check the quality of the approximation of the estimators by \textsf{nlm}, we simulate variable sizes of samples distributed according to the censored Gaussian model. The values of these estimations are plotted in SI~B.3.

\subsection{Censored mixture model}
\label{Sec:mixing:censure}

We apply here the statistical analysis described in the previous section to simulated data based on the values for the viral load distribution found in~\cite{Jones} with a mixture model and a censoring threshold $\dcens\approx35.6$ (so the two rightmost bars in the histogram of Fig.~\ref{fig:distribution_viral_load}, that appear much smaller than the nearby values, are supposed to be censored). It is reasonable to assume that the censoring threshold has the same value for each sub-population, as it depends on the test methodology rather than on the tested individuals. In Fig.~\ref{fig:distribution_viral_load}, we represent the histogram with the density for the mixture.

We observe that the separation in sub-populations and the resulting densities are very close to the ones obtained in the ``naive'' Gaussian mixture model, constructed without taking into account the detection threshold. The principal difference between the naive and  censored models consists, for the later, in a larger variance that extends above the threshold. To a lesser extent, the sub-population with a median concentration can also exceed the threshold. It is worth mentioning that as expected, the probability of detection below the threshold value is sensibly the same for all three clusters (around 20\%).

As a result, using the computed estimates (see Tab.~III 
) and the model, we can calculate a theoretical false negative rate, see SI. Eq. [S3]: in this case, the value is approximately $3.8\%$ (represented by the red area on the Fig.~\ref{fig:distribution_viral_load}b); it mostly belongs to the third cluster. Besides being based on simulated data, hence subject to caution, this value should be treated as a lower bound, as it is possible that a fourth cluster of infected individuals exists but with very small viral load, below the detection threshold. It allows us to predict for example that at least 150 clinical tests in the series of~\cite{Jones} might have resulted in false negatives due to their low viral load.

To validate the censored model, we can verify that if one (i) erases the data to the right of a certain value and (ii) uses the totally censored model on the remaining data, a similar estimate should be obtained for the parameters. We display in SI. Fig.~S7 
the density obtained using the censored mixture estimation with $\dcens\approx35.6, \ 34.4$ and $33.2$ (removing the first two, the third, then the fourth rightmost bars in the histogram). We observe that the first and second components are globally unchanged. The mean and standard deviation of the last component are almost the same for $\dcens\approx34.4$ and $\dcens\approx35.6$ (see SI. Tab. IV 
); only the proportions naturally decrease with the threshold. On the other hand, the mean moves slightly to the left for $\dcens\approx33.2$; this is due to the fact that we loose the information of the largest bars of this component. It might also be caused by our ignorance of the exact distribution of $C_t$ values within classes of the histogram (we recall that we assume that it is a uniform distribution).

Note that if we were to set the threshold at $\dcens \approx 34.4$ as threshold for the partially censored model without erasing data, the optimization procedure \texttt{nlm} would not converge. This is further indication that a detection drop happens in the neighbourhood of $35.6$.

\subsection{Application to other datasets}
We applied a similar statistical analysis to simulated datasets of 852 infected nursing home resident and workers studied by~\cite{Cabrera2020}. For this dataset, two to three sub-populations are identified by our algorithm. The estimation obtained for the Gaussian fit of the $C_t$ distribution they obtained is given in SI Tab. \ref{tab:MM_Gaussian}. It would be interesting to link these observed sub-populations to characteristics of the individuals (e.g. age of the patients or stage of the disease). We mention that in two other datasets of smaller size~\cite{Yelin2020,Liu2020}, the statistical resolution does not allow us to distinguish between several sub-populations; we rather found that the distribution of $C_t$ corresponds to a single Gaussian with standard deviation $\sigma$ in the $5$ to $6$ range.

\section{Group testing: application to epidemiology}
\label{sec:epidemiology}

We now show how the previous analysis of the tests used to measure the viral load in patients can be used to precise the epidemiological monitoring of the disease in the general population.

\subsection{The number of infected individuals in pools cannot be recovered}

We expect the PCR result to correspond to the sample with the highest viral load, up to a dilution-induced drift $\log_2(N)$, under the model hypothesis of Sec. \ref{sec:modeling} (cf. Fig.~\ref{fig:PCRdescription}). Indeed, since the viral concentration in randomly selected infected individuals spans several order of magnitudes, we expect that
\begin{align}
\log_2\left(\frac{1}{N} \sum_{i=1,\ldots,j} C_i\right) \approx \log_2\left(\max_{i=1,\ldots,j} C_i\right) - \log_2(N), \label{eq:min_viralload}
\end{align}
for $j$ positive samples with concentration $C_j$ diluted in a pool of $N$.

We point out that the RT-qPCR viral load measure could be used to improve efficiency and cross testing of smart pooling type diagnostic methods, which are beyond the scope of this paper. We plan to investigate this aspect in future work.

Therefore, the measured value of the pooled sample viral concentration cannot be used to estimate the number of infected individual within the pool. 

\subsection{Group testing and the measure of viral prevalence}
Here, in contrast with Sec.~\ref{sec:sc}, we no longer consider that the RT-qPCR tests to be perfect. As discussed in \refn{eqn:groupModel}, we model the concentration of the pooled sample as the average of the individual sample loads; and we assume that viral concentration becomes undetectable below a given threshold. Therefore, creating groups has the effect of increasing the false negative rate, which has to be quantified. We then use this estimation to un-bias the estimator of the prevalence in the overall population based on group testing, and study its impact on the optimal choice of group sizes.

\subsubsection{Estimation of the false negative rate induced by pooling}

The distribution of the viral load of a single positive sample within a pool of several negative samples appears as shifted towards higher $C_t$-values, see Fig. \ref{fig:PCRdescription}. A pooled sample returns positive only if the average concentration is smaller than $\dcens$; thus using the observation of Sec.~\ref{Sec:mixing:censure}, contamination will be detected in a group of $N$ individuals typically if at least one individual in the group has a viral charge larger than $N 2^{-\dcens}$. Therefore, there is a risk that low viral charge samples (that would have been tested positive using individual tests) would no longer be positive in pool tests. Similarly to \refn{eq:assumption}, we express the increased rate of false negative due to pooling as $\P(- \log_2(C) + \epsilon_2  \geq \dcens - \log_2(N)), $
where $\log_2(C)$ is the viral concentration of the positive individual. For simplicity we neglect the measurement error of the qPCR, i.e. considering that $\rho = 0$, thus an expression for the increased rate of false negatives reads $1-\Phi(\dcens^{(N)})$, where
\begin{equation}
  \Phi\left(\dcens^{(N)}\right) = \P\left[- \log_2(C)  \leq \dcens - \log_2(N)\right].
  \label{eq:falsenegativerisk}
\end{equation}
We find that, when estimated by the censored model, the false negative risk function $\Phi(\dcens^{(N)})$ grows quicker as the pool size increases than in the uncensored model, see Fig.~\ref{fig_falsenegativerisks}. This is partly due to the fact that the censored model makes the assumption that $\dcens \approx 35.6$, whereas $\dcens \approx 37.3$ in the uncensored model. Choosing a correct statistical model for the distribution of $C_t$ values has a critical impact on the estimation of the false negative risk, therefore on the estimate of the prevalence, as is discussed next.

%
%
%
%
%
%
\begin{figure}[t!]
\centering
\includegraphics[width=8.5cm]{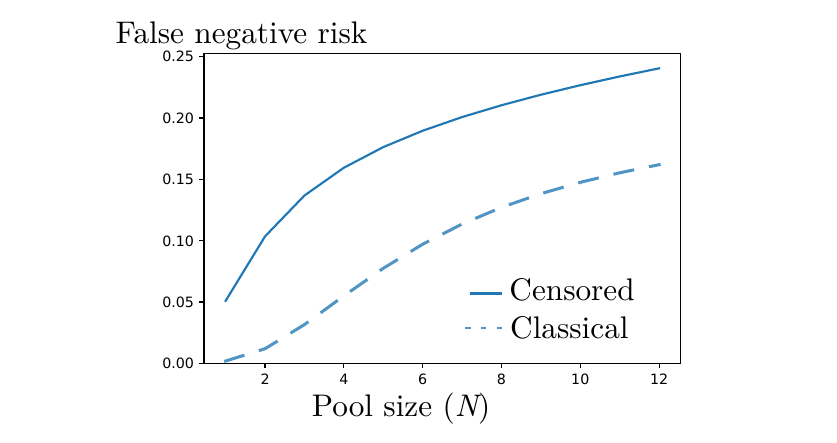}
\caption{Estimated false negative risk  probability of a sample pooling test containing a single infected individual (with a viral load distributed according to the simulated statistics presented in Sec. \ref{sec:statistical}) as a function of the number of pooled individuals $N$, using the estimated density obtained for the classical (dashed line) and censored (solid line). The censored and classical model give different estimates of (and probably both underestimate, as we did not consider unsuccessful sampling) the false negative rate in group pooling under the model assumption of Eq. \eqref{eqn:melange}; such estimates will have impact on the estimation of the prevalence.}
\label{fig_falsenegativerisks}
\end{figure}

\subsubsection{Correction of the prevalence estimate by the false negative rate}

Assuming a false negative rate of $1 - \Phi(\dcens^{(N)})$ in pool testing with groups of size $N$, we observe that $1 - (1 - \bar{X}^{(N)}_n)^{1/N}$ (as defined  using the notation of Section~\ref{sec:sc}) is a consistent estimator of $p \Phi(\dcens^{(N)})$ (c.f. Lemma~\ref{lemma1}). As a result, the confidence interval constructed for the prevalence $p$ now reads
\begin{align}
  \mathrm{CI}_{1-\alpha} &(p) =  \left[ \frac{1 - (1 - \bar{X}^{(N)}_n)^{1/N}}{\Phi(\dcens^{(N)})} \right. \nonumber \\
  & \left. \pm \frac{q_\alpha}{\sqrt{n}} \frac{(1 - \bar{X}^{(N)}_n)^{1/N - 1}{\sqrt{\bar{X}_n^{(N)} \left(1 - \bar{X}^{(N)}_n\right)}}}{N\Phi(\dcens^{(N)})} \right].
  \label{eq:confidenceinterval}
\end{align}

For the numerical applications presented in  Fig.~\ref{fig_prevalence} and Table~\ref{tab:title}, we consider a viral load $C$ that is distributed according to ~\refn{eq:viral_load_distribution}. As expected, due to false negatives, we find that the number of tests needed to reach a given precision on the prevalence is increased, but relatively moderately. 
In particular, the optimal pool size value, $N^{\imp}_\text{opt}$, that minimizes the number of tests needed to reach a given precision level, is close to the value $N^{\perf}_\text{opt}$, defined in \refn{eqn:nopt}. 

\begin {table}[t!]
{
\begin{center}
\caption {Table of the pool size as a function of the number of tests for a prevalence of $3\%$ measured with a precision of $0.2\%$ at a $95\%$ confidence interval, for both perfect tests (with no false negatives, see Sec.~\ref{sec:sc}) and imperfect tests (with false negatives; model parameters defined in Table~\ref{tab:monitoring}); computed using Eqs.~[\ref{eqn:ci}] and [\ref{eq:confidenceinterval}].} \label{tab:title} 
\begin{tabularx}{8.5cm}{|m{1cm}|X|X|X|X|}
\hline
Pool & \multicolumn{2}{c|}{Perfect tests}& \multicolumn{2}{c|}{Imperfect tests} \\
\cline{2-5}
size $N$&Number of tests $n$& Sample size $n N$&Number of tests $n$& Sample size  $n N$\\
\hline
\hline
1&{29100}&29100&{29464}&29464\\\hline
2&{14775}&29550&{15069}&30138\\\hline
3&{10003}&30009&{10261}&30783\\\hline
5&{6191}&30955&{6411}&32055\\\hline
10&{3350}&33500&{3530}&35300\\\hline
20&{1973}&39460&{2130}&42600\\\hline
30&{1561}&46830&{1716}&51480\\\hline
50&{1349}&67450&{1525}&76250\\\hline
100&{1884}&188400&{2235}&223500\\\hline
200&{10378}&2075600&{13105}&2621000\\\hline
\end{tabularx}
\end{center}
}
\end {table}

\begin{figure}[t!]
\noindent\fcolorbox{black}{blanchedalmond}{
\begin{minipage}{8cm}
\begin{center}
{\color{blue}{\textbf{Box 2: A protocol of prevalence determination}} \vskip-0.3cm}
\end{center} 
\justify
\noindent   
We propose the following procedure for the measure of prevalence via group testing:
\begin{enumerate}[leftmargin=*]
  \item Start from an a priori estimate for the prevalence ($\hat{p}_0$).
   \item Based on the value of $\hat{p}_0$, estimate the number $N$ of individuals in the pool that minimizes the total number of tests needed to achieve the estimation of the prevalence $p$ at the targeted precision and confidence interval.
  \item Construct a number of $n$ pools containing each $N$ individuals selected at random in the general population, with $n$ the number of tests available for the measure.
  \item Count the number of positive tests and compute the average  $\bar{X}_n^{(N)}$.
  \item An improved estimate of the prevalence then reads: $\hat{p}_1 = 1 - (1 - \bar{X}^{(N)}_n)^{1/N}$ (cf Lemma~\ref{lemma1}).
\end{enumerate}
Note that this method can easily be adapted into a Bayesian algorithm, with the number $N$ of individuals tested modified at each iteration of the procedure.
\end{minipage}}
\label{box2}
\end{figure}

Similarly, one can observe that using a different distribution with similar mean and variance for $-\log_2 C$ as~\refn{eq:viral_load_distribution} would lead to moderate changes of the values estimated in Table~\ref{tab:title}. While modelling of the viral load of an infected individual is crucial to un-bias the estimator of the prevalence via group testing, the practical implementation of such group testing strategy, i.e. the choice of the group size $N$ and the number $n$ of tests to use, is relatively independent of the precise statistical properties of the viral load distribution. 

Based on \refn{eq:confidenceinterval},  in Box~2. we propose an iterative method to estimate $p$, which, during a survey, allows for  on-the-fly adaptations of the pool size. 

\subsection{Group testing and Bayesian inference of the prevalence in sub-categories of the population}

The viral prevalence may vary significantly among specific categories within the overall population. In particular, a prevalence reaching $5\%$ was measured among the health care workers population in a hospital~\cite{Barrett2020}, which we expect to be significantly higher than the estimate prevalence within the general population.

Here we show that we do not specifically need pool samples from individuals from homogeneous categories in order to recover the distribution of prevalence within these categories. 

The protocol described in Box~2 can be adapted to study different prevalences in specific sub-populations, provided that the number of individuals of each subpopulation is known for every grouped sample. In SI. Fig.~\ref{fig:Bayesian_twopopulationsSI}
, we evaluate, as function of the number of tests, the credibility intervals on the prevalence within two categories of the population: one at $p_1 = 5\%$ representing $20\%$ of the total population (a value inspired by \cite{Barrett2020}), the other being at $p_2 = 0.5\%$ (a value inspired by \cite{Gudbjartsson2020}). More information on this adaptative protocol can be found in the SI. 

\begin{remark}
Note that once a difference in prevalence is noted from the epidemiological study of the general population, testing can be adapted to construct groups containing only members of one subpopulation to attain similar levels of precision for the prevalence of the sub-populations. The prevalence in the general population can then be recovered by averaging the estimators of the sub-populations. The advantage of these adaptative settings is that the existence of a difference of prevalence in populations can be tested before deployment of resources needed to measure them specifically.
\end{remark}


\subsection{Group testing for regular monitoring in communities}

We now consider some applications of group testing to the early detection of an epidemic outbreak within a community, that is interconnected and reasonably closed to the outside (e.g. schools, nursing homes, detention centres). 


Our results are based on a minimal working model for epidemic outbreaks within a community of $A$ individuals. At a random date, which we choose to be time $t = 0$, the patient $0$ in the population gets contaminated, and immediately starts infecting members of its community at rate $\lambda$. Each newly infected individual then contributes to spreading the disease at the same rate $\lambda$. 

\subsubsection{Optimization of the size of pools} \label{sec:poolsize} We first consider the impact of group testing strategy, consisting in $k$ group test with pools of $N$ individuals, on the early time of the outbreak $t \ll \lambda^{-1}$. With a unique contaminated individual in the population, the detection probability reads
\begin{equation}
 \P\left[\mathrm{+} \lvert k\ \mathrm{tests}\right] = k N \Phi_0(\dcens^{(N)})/A, \quad \mathrm{with}  \ k N \leq A,
 \label{eq:testcommunities}
\end{equation}
where $\Phi_0(\dcens^{(N)})$ is defined according to \refn{eq:falsenegativerisk}, with the difference that the assumed viral load of the patient $0$, corresponding to that measured at early times, may need not be equal to the distribution estimated in \refn{eq:viral_load_distribution} based on clinical data. For simplicity, we will assume in the following that $\Phi_0$ is the cumulative distribution of a  log-normal viral load distribution $\mathrm{logN}(\mu_0,\sigma_0)$ of mean $\mu_0$ and variance $\sigma_0$. 

In Fig.~\ref{fig_detection}, we represent the evolution of the probability to detect the patient $0$ as a function of the total number of sampled individuals in a population of size $A = 120$. We observe that if $\mu_0$ is close enough to $\dcens$, i.e. if the viral charge of the patient 0 is close to being undetectable, then there will exist an optimal size for the pools. When $N$ becomes too large the risk of false negative overcomes the potential benefits of testing larger portions of the community (see Fig.~\ref{fig_detection}a). In contrast, if the viral load of patient $0$ is slightly higher, the detection probability becomes a monotonic function of the pool size $N$, indicating that larger pools are always beneficial. Additionally, if using multiple tests increases the detection probability when the viral load is close to the detection threshold, using multiple tests has a smaller impact when the viral load gets easier to detect.

\begin{figure}[t!]
\centering
\includegraphics[width=8.5cm]{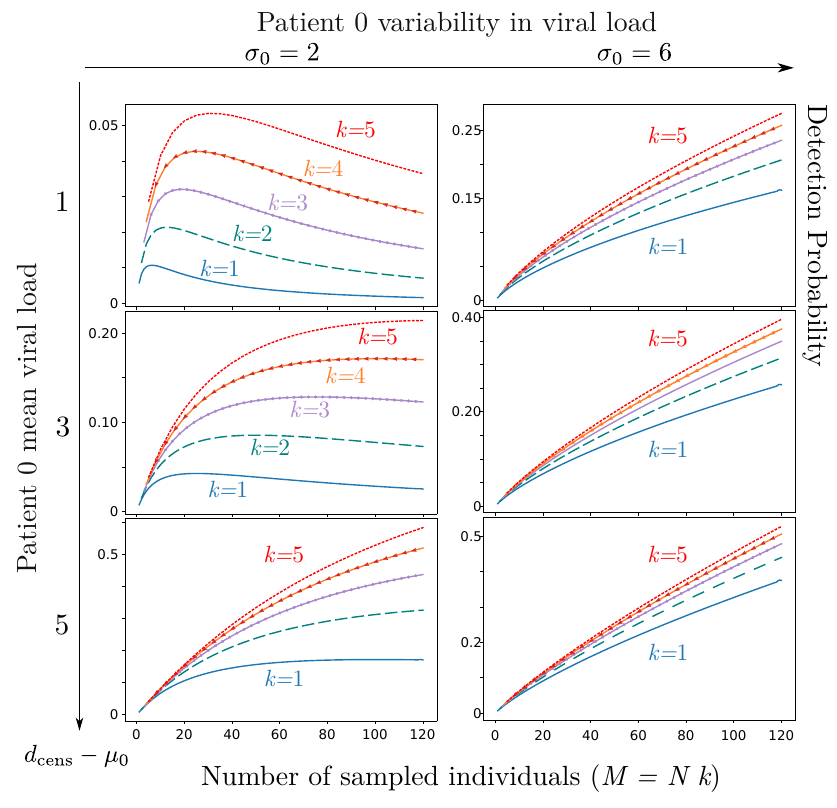}
\caption{Detection probability of a single patient $0$ with low viral load within a community of $120$ as a function of the total number of sampled individuals $M = k \times N$, where $k$ is the total number of tests used and $N$ the number of samples pooled together in a test, with $k=5$ (red dotted line); $k=4$ (orange line with arrow), $k=3$ (purple line with circles); $k=2$ (dashed green line); $k=1$ (solid blue line) for several values in the parameters describing the viral load of the patient $0$ at the onset of contagiosity, expressed in terms of a normal distribution in $C_t$ (the number of RT-qPCR amplification cycles) with a standard deviation $\sigma_0$  and a mean $\mu_0$ and a threshold at a value denoted $d_{\mathrm{cens}}$ satisfying: $\mu_0 =d_{\mathrm{cens}} -1$ (top row), modelling a patient 0 with a very low viral concentration, $\mu_0 = d_{\mathrm{cens}}-3$  (middle row), $\mu_0 = d_{\mathrm{cens}} - 5$ (bottom row); $\sigma_0 = 2$ (left column); $\sigma_0 = 6$ (right column). }
\label{fig_detection}
\end{figure}

\begin {table}[ht!]
{
\begin{center}
\caption {Table with standard parameter values (with std. the abbreviation of standard deviation).} \label{tab:monitoring} 
\begin{tabular}{|l|l|l|}
\hline
\text{Symbol} & \text{Meaning} & \text{Value}\\
\hline
\hline
$\dcens$ & Maximal cycle number 			 	& $35.6$ \\
$\mu_i,\sigma_i,p_i$ 	& Viral load (in $C_t$) distribution fits			& Table~IV 
 \\
$\rho$ 				& PCR measurement error (std.)			& $0$   \\
\hline
\hline
$\phi$ 				& Delay before onset of symptoms 		& $5 \, \mathrm{days}$ \\
$\lambda$ 		& Intra-community contamination rate 		 		& $0.5 \, \mathrm{days}^{-1}$ \\
$r$ 					& Asymptomatic probability 					& $40\, \%$\\
$\tau$ 				& Time interval between grouped  tests & $1-12 \, \mathrm{days}$\\
$A$ 					& Total number in the community			& $120$ or $1000$  \\
$N$ 					& Pool size 											& $1$--$128$ \\
$\mu_0$ ; $\sigma_0$ 	& Viral load (in $C_t$) (mean, std.)	& $30-35$  \\
$\zeta$ 	& Defective sampling probability	& $0$    \\
\hline
\end{tabular}
\end{center}
}
\end {table}

\subsubsection{Optimization of the regularity of tests}


We now consider the impact of a regular testing strategy every $\tau$ units of time, with $k$ pools of $N$ individuals with individuals drawn at random in the population. To compare the impact of the frequency of testing for a constant budget, we assume that the ratio $k/\tau$ is equal to a fixed constant, representing the amount of tests spent per unit of time by the community to detect infections. What we consider here is a continuous-time steady-state version of the previous paragraph \ref{sec:poolsize}.

We denote by $r$ the proportion of asymptomatic individuals. Symptomatic individuals start showing signs of being contaminated a given number of days after infection, denoted by $\phi$ (that we assume to be constant in this simplified model). The outbreak is detected either if one of the screening tests finds a contaminated individual, or if one of the contaminated individuals shows symptoms. We denote by $T_s$ the epidemic detection time without screening tests (as defined as the first detection of a symptomatic individual) and $T_d$ is the one through random tests screenings  time, see Fig. \ref{fig_regularity}.


In this model, the number of infected individuals at the date $t$, denoted by $N(t)$, is distributed according to a geometric random variable with parameter $1-e^{-\lambda t}$, as expected for such a Yule process~\cite{Meleard2016}. In the absence of screening tests, we find that the average detection time reads $\langle T_s \rangle = \phi - \log (1-r)/\lambda$. The average number of contaminated individuals at that time then reads: $\left\langle N(T_s ) \right\rangle  = e^{\lambda\phi}/(1-r)$. In comparison, based on~\refn{eq:testcommunities}, we find that the first screening test after contamination will detect the outbreak with a probability 
\begin{equation}
\P\left[\mathrm{+} \lvert k\ \mathrm{tests}\right] = \frac{k (e^{\lambda \tau}-1) \Phi(\dcens^{(N)})}{\lambda \tau A}.
 \label{eq:continuous_time}
\end{equation}
By the time of detection $t=\tau$, a number of  $\left\langle N(\tau) \right\rangle  =(e^{\lambda \tau} - 1)/\lambda \tau$ individuals has been contaminated.

\begin{figure}[t!]
\centering
\includegraphics[width=8.5cm]{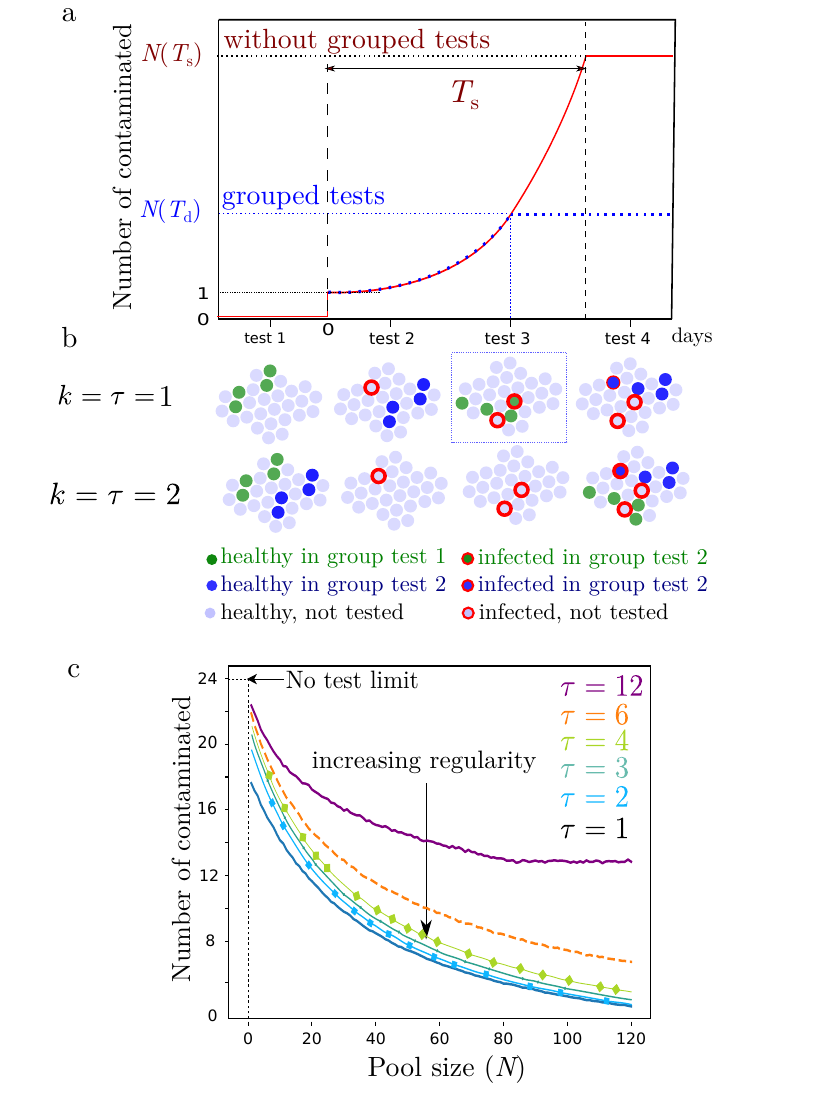}
\caption{(a) Sketch of the time evolution of the number of contaminated individuals in a community. The patient $0$ is contaminated from the outside of the community 
$0.8$ units of time after a test date. In the absence of screening tests, the contamination is detected at the time $T=T_s$ (after appearance of the first symptoms); with grouped tests, an infected individual is detected at a time $T=T_d$.
(b) Sketch of two group testing strategies, here with pools of size $N=4$, one with a single ($k = 1$) grouped tests every day ($\tau = 1$); the other with $k = 2$ grouped tests every second day ($\tau = 2$); the second strategy (least frequent testing) fails to detect the outbreak early and results in more contamination. (c) Number of infected individuals at the detection of the outbreak as a function of the pool size, using $k=\tau$ tests performed at $\tau$-day intervals, with $\tau = 12$ (solid purple line), $\tau = 6$ (dashed orange line), $\tau = 4$ (dark green solid line with square), $\tau = 3$ (light green solid line with circles), $\tau = 2$ (cyan line with squares) and $\tau = 1$ (solid blue line). Here we consider a large community composed of $A = 1000$ individuals.  The patient 0 has a viral load concentration distributed according to a log-normal distribution with mean $\mu_0 = 30$ and standard deviation $\sigma_0 = 2 \log_2 (2)$; all others parameters can be found in Table II.}
\label{fig_regularity}
\end{figure}

Computing the average number of contaminated individuals $\langle N(T_d) \rangle$, we find that a screening strategy consisting in sampling a random subgroup of the community as frequently as possible is more efficient than the one consisting in testing larger portion of the community at less frequent time intervals. In Fig.~\ref{fig_regularity}, we compare different screening scenarios for a large community composed of $A = 1000$ individuals. We vary the value of the screening time interval $\tau$ while keeping fixed (1) the average number of tests per unit of time and (2) the size of the pools on which each test is used. Our simulation range from checking $N$ individuals every day (with one test) to checking $12 \times N$ individuals in $12$ pools every $12$ days. 

\subsubsection{Discussion}
In Fig. \ref{fig_regularity}, we chose to consider the case of presymptomatic patient 0 with a weak viral load ($\mu_0 \approx 30$); we find that a pooling method remains efficient for monitoring purposes even in such unfavorable settings. The recent study \cite{Arons2020} indicates that a large fraction of presymptomatic individuals detected in a nursing homes had relatively high viral loads (in the $\mu_0 \approx 20-25$ range), which tends to indicate that screening methods based on pooling would be even more efficient than suggested in Figs. \ref{fig_detection} and \ref{fig_regularity}.


We also point out that, here, we have used a minimal set of parameters in order for analytical calculations to be tractable. Including more parameters (e.g. considering a time-dependent infection rate or viral charge for patients after their contamination, graph of relationship within the community) would be needed in order to obtain conclusive results to be used as healthcare guidelines. In this direction, based on stochastic simulations encompassing a large set of parameters, \cite{Smith2020} also concludes on the efficiency of group testing in preventing epidemic outbreaks in health care structures. We hope our analytical calculations will inspire further epidemiological investigations. Without such further studies, our results should not  be used literally as a guideline for clinical practice/healthcare related behaviour.

\section*{Conclusion}


We consider the effect of sample dilution in  RT-qPCR grouped tests and we propose a model to describe the risk of false negatives as a function of the pool size. We present a procedure to analyse experimental datasets for the viral charge of patients. Inspired by the clinical study~\cite{Jones}, we expect the statistics of the number of amplification cycles to be well described as a mixture of 3 Gaussian variables censored at the RT-qPCR sensibility limit. Our analysis may hint at the existence of 3 sub-classes of interest that could each be interpreted in terms of medical or physiological criteria~\cite{Liu2020}. 

We point out that the viral load distribution that we determine here based on the clinical sampling presented in~\cite{Jones} is possibly biased as compared to the viral load that would be measured within a population targeted by a group testing strategy. Indeed, asymptomatic individuals may exhibit a different viral load distribution than the population that got tested. Furthermore, the measures of~\cite{Jones} were based on nasal swabs samples, and their statistical distribution is likely to vary according to the sampling method chosen for group testing.

Tests based on saliva samples appear to show promising results in terms of false negative risks while improving the acceptance of tests among the population, decreasing the time needed for sample collection and reducing the exposure of health care practitioners~\cite{Wyllie2020,Azzi2020,To2020,Williams2020,Khurshid2020}. In this context, we are looking forward to seeing whether group testing on saliva samples would provide reliable results;  statistics on the distribution of the viral load in saliva samples would then be of interest to evaluate the efficiency of a group testing strategy based on such samples.

We think group testing could provide the means for regular and massive screenings allowing the early detection of asymptomatic and pre-symptomatic individuals -- a particularly crucial task to succeed in the containment and potential eradication of the epidemic~\cite{Treibel2020,Kissler2020, Ferretti2020}.

\section*{Acknowledgments} We wish to thank members of the MODCOV initiative and in particular Françoise Praz and Florence Debarre who gave us numerous helpful comments, including on the first version of this manuscript. We thank Philippe Hupe for his critical reading. We also thank Marie-Claude Potier and Marc Sanson for insightful discussions on  RT-qPCR tests as well as on the interest of saliva samples.


%


\pagebreak

\clearpage

\onecolumngrid

\appendix

{\begin{center} {\normalsize{\textbf{Supplementary Information}}} \end{center}}
{\begin{center} {\normalsize{\titlename}} \end{center}}
{\begin{center} {\normalsize{Vincent Brault, Bastien Mallein, Jean-Francois Rupprecht}} \end{center}}

\setcounter{figure}{0}
\renewcommand{\thefigure}{S\arabic{figure}}
\setcounter{equation}{0}
\renewcommand{\theequation}{S\arabic{equation}}

\section{Ideal tests}

We present here some of results obtained from the computations made in Sec.~I, where we assumed perfect group testing and used it to measure prevalence in the population. Note that with a perfect test, the question of early detection of an outbreak in a community becomes much simpler : one just need to test everyone at regular time intervals with a single test.

\subsection{Number of tests and sample size as function of the population prevalence}

We trace here, for various values of the prevalence, the number of tests and total number of samples needed to archive a given precision for the confidence interval.

\begin{figure}[ht!]
\centering
\includegraphics[width=10cm]{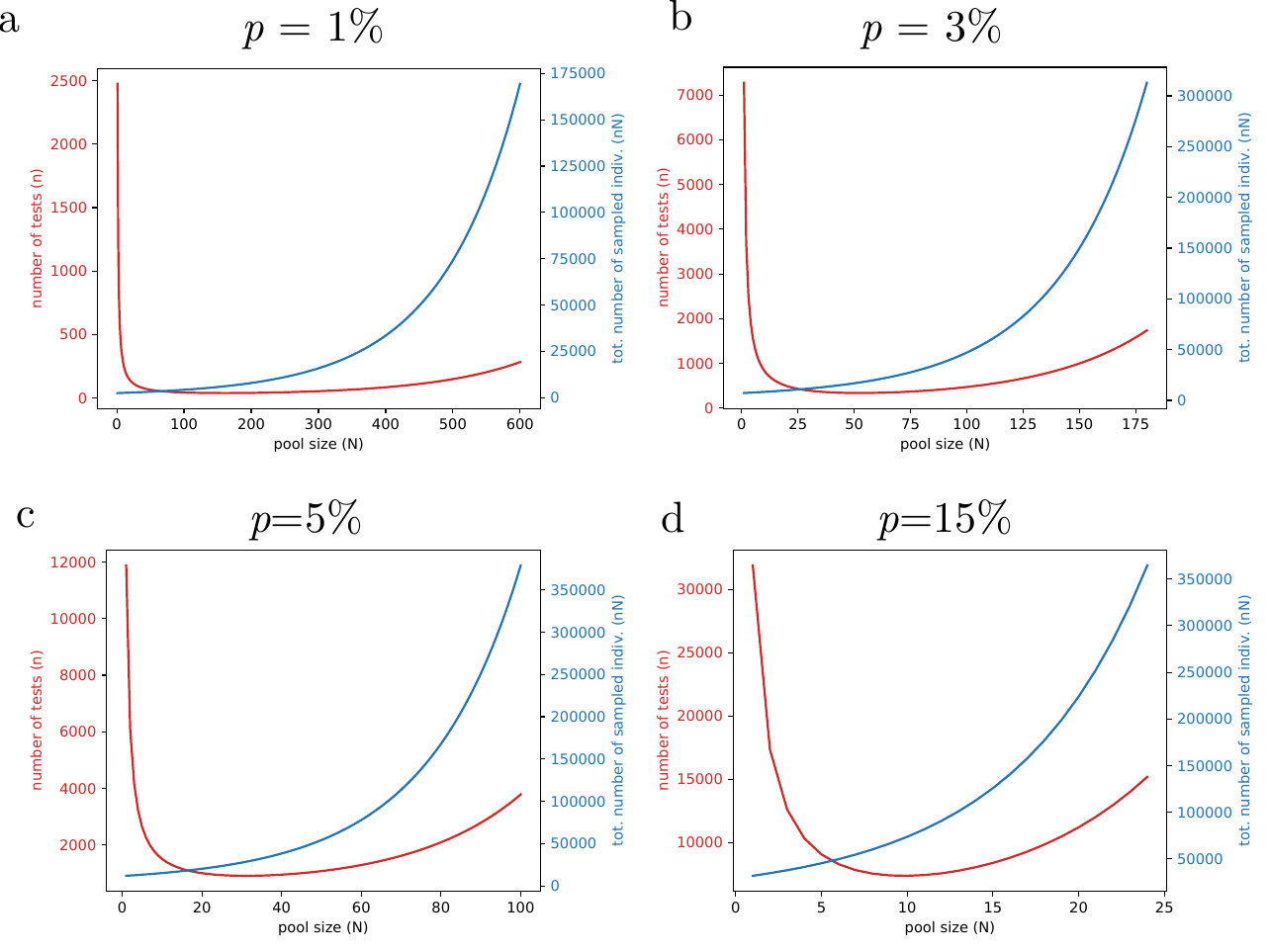}
\caption{Total number of tests and sampled individuals so that the width of the $95\%$ confidence interval is smaller than $0.4\%$ as a function of the pool size $N$ chosen for a perfect test, for a prevalence $p$ equal to $p =1\%$ (a), $p =3\%$ (b), $p =5\%$ (c), $p =15\%$ (d).}
\label{SI:fig_prevalence}
\end{figure}

\subsection{Bayesian inference}
\label{subsec:si:bi}

We are now interested in a Bayesian approach to the measure of prevalence. We started with an initial prior distribution with density $f_0(p) = 6p(1-p) \ind{0 \leq p \leq 1}$ for the prevalence, and for each new test $j$ we do the following:
\begin{enumerate}
  \item take the the mean value $\bar{p}_{j-1} = \int_0^1 p f_{j-1}(p)\mathrm{d} p$ of the prior;
  \item choose the size $N_{j}$ of the pool of the $j$th test computed as (cf.~\refn{eqn:nopt}): 
\begin{align}
N_{j} = \floor{-\frac{c_\star}{\log(1-\bar{p}_{j-1})}};
\end{align}
  \item choose $N_j$ individuals at random and test them in a group:
  \begin{itemize}
    \item if the test is positive, then $f_j(p) = C^+_j (1 - (1-p)^{N_j}) f_{j-1}(p)$;
    \item if the test is negative, then $f_j(p) = C^-_j (1-p)^{N_j} f_{j-1}(p)$;
  \end{itemize}
  with $C_j^\pm$ normalizing constants, chosen such that $\int_0^1 f_j(p) \mathrm{d} p = 1$.
\end{enumerate}
We trace in Fig.~\ref{fig:Bayesian_twopopulations} the result in blue of this experiment, the $95\%$ credible interval being $[a_j,b_j]$, with $a_{j}$ being the $2.5\%$th quantile of $f_j$ and $b_j$ its $97.5\%$ quantile.

Simultaneously to this statistical experiment, one can follow the prevalence in sub-populations of interest. For example, if we assume the population consists of two sub-populations $1$ and $2$ with different prevalences $p_1$ and $p_2$. Starting with a prior distribution $f_j(p_1,p_2) \mathrm{d}p_1\mathrm{d}p_2$ for these prevalences, if a group consisting of $a$ individuals of the first sub-population and $b$ individuals of the second population is sampled positive, then Bayes rules gives $C_{j+1}^+ (1 - (1-p_1)^a(1-p_2)^b)$ for the updated law of $(p_1,p_2)$. A similar update is made if the test is negative. As a result, we get estimates for the prevalence in each sub-population at the same time as we are measuring the prevalence in the overall population.

We test the above statistical experiment on a population  which is composed of two sub-po\-pu\-la\-tions, one large subpopulation of sparsely exposed individuals (prevalence $0.5\%$, representing $4/5$th of the whole population), and a smaller subpopulation of very exposed individuals (prevalence $5\%$). At each step, we choose the size of the pool according to the available estimate for the prevalence in the complete population. The composition of the pool in terms of individuals of each sub-population is chosen at random (at the $j$th step, there are $\mathrm{Ber}(N_j,0.8)$ individuals of the first sub-population). We also update our estimation of the prevalences $(p_1,p_2)$ in each of the two sub-populations.

The results are traced in Fig.~\ref{fig:Bayesian_twopopulations} in orange and green curves. One can see that the width of the credibility intervals of the sub-populations decay much slower than for the whole population. The reason is that the size of the groups are optimized to measure as precisely as possible the mean value $p$.

\begin{figure}[t!]
\centering
\includegraphics[width=8.5cm]{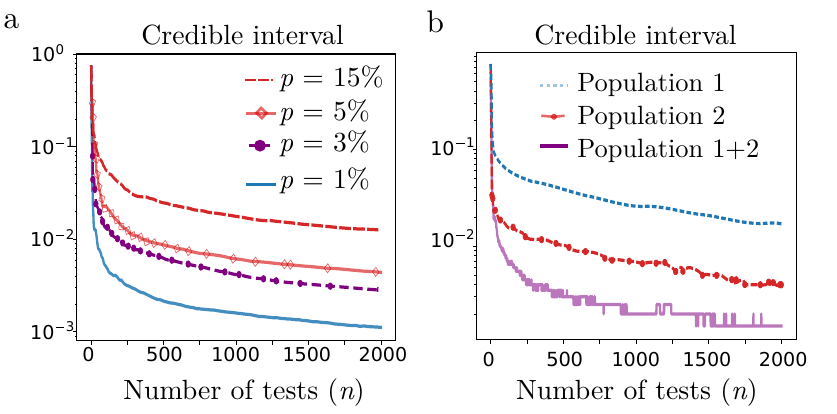}
\caption{(a) Width of the $95\%$ credible interval on the prevalance $p$ with adaptative Bayesian sampling as a function of the number of tests $n$ for a set of values in the prevalence ranging from $p=15\%$ (top, magenta dashed line) to $p=1\%$ (bottom, blue solid line).
(b) Width of the credible intervals in a two-category mixed population for the prevalence either: in the general population (magenta solid line); for the less exposed population $1$ with a prevalence of $0.5\%$, representing $80\%$ of the general population  (blue dashed line); for the more at-risk population $2$ with a prevalence of $5\%$ representing $20\%$ of the general population (red dotted line with circles). }
\label{fig:Bayesian_twopopulations}
\end{figure}

\begin{figure}[ht!]
\centering
\includegraphics[width=10cm]{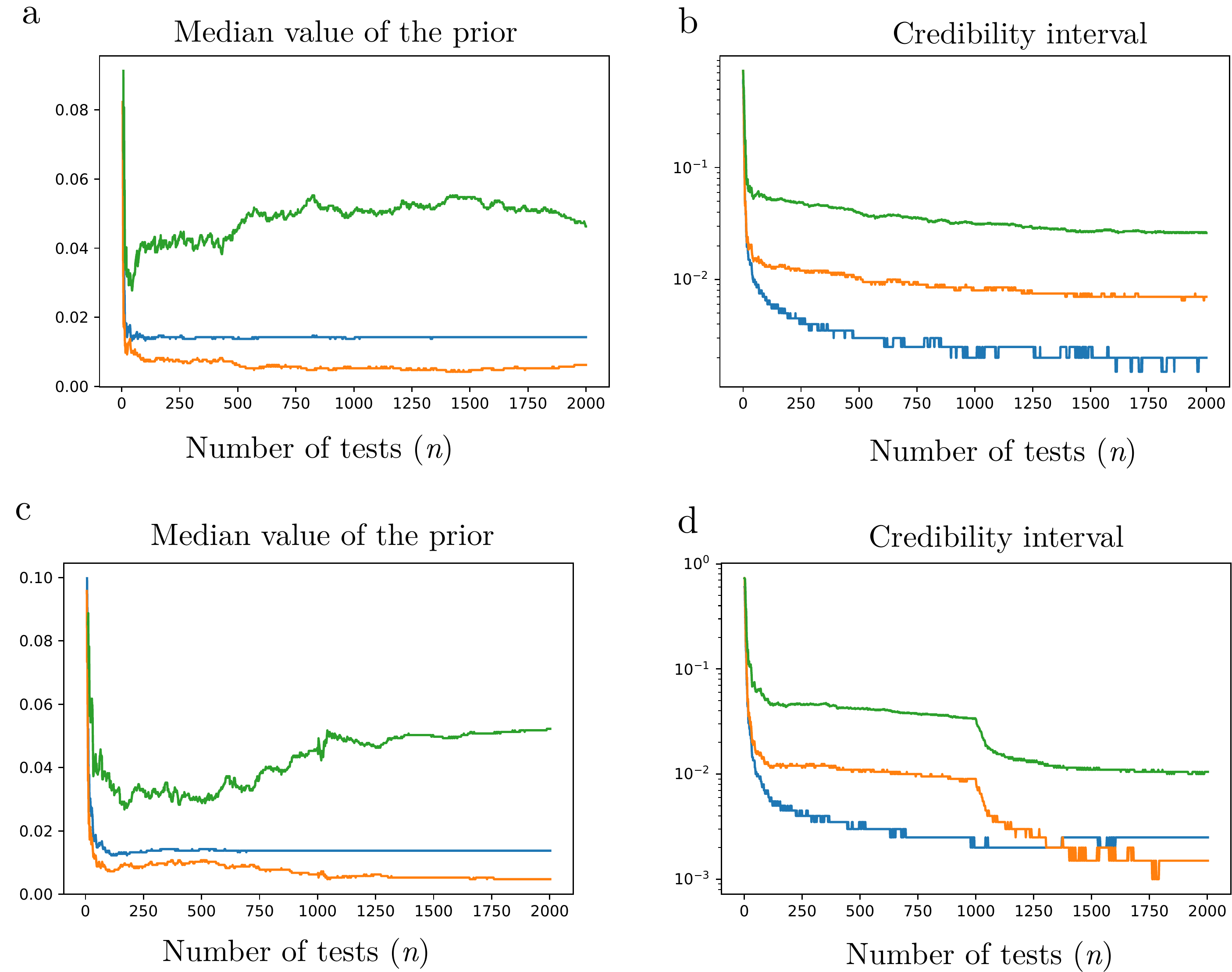}
\caption{(a-b) Bayesian estimation of the parameters of a mixed population, consisting of $80\%$ individuals of type $1$ with a prevalence of $0.5\%$ and $20\%$ individuals of type $2$ with a prevalence of $5\%$. Pooled samples are constituted by sampling randomly individuals from the two sub-populations, with a size optimized for the speed of convergence of the overall prevalence of $1.4\%$. (a-c) Median value of the priors, overall population in blue, first resp second population in orange resp green. (b-d) Width of the $95\%$ credible intervals.
In (c,d), the first 1000 tests are made on groups whose size is optimized to estimate the prevalence of the overall population, the next 1000 tests are divided into two groups that are used on homogeneous sets of the sub-populations, in groups optimized to estimate the prevalences within these sub-populations. This has the effect of drastically improving the speed of convergence of the estimator of the prevalence in the subpopulations.
}
\label{fig:Bayesian_twopopulationsSI}
\end{figure}

However, observe that even with a naive group construction (without segregating individuals according to their sub-population), one can extract information on the prevalence of the sub-populations of interest. Therefore, a design for the measure of the prevalence in a stratified population could be the following: in a first time, pool testing is implemented on randomly constructed group of individuals from the general population. Data is then analysed to detect sub-populations with different prevalences (e.g. according to geography, age, occupation, ...). In a second time, once sub-populations of interest are identified, pool testing is applied to each of the sub-populations independently. We implemented this method if Fig.~\ref{fig:Bayesian_twopopulationsSI}, with the same number of tests a much more detailed estimate of the prevalence is obtained.


%

\newpage

\section{Censored Gaussians}

In this section, we present the simple mixing models of Sec.~\ref{sec:statistical}, as well as some complementary graphs and the estimations obtained for the parameters of this models.

\subsection{Naive method based on mixing models}

In this section, we trace the density estimated by a simple mixture of Gaussian variable presented in Sec.~\ref{sec:Naive:MM}. An estimation of the parameters of this mixture are given in Table~\ref{tab:MM_Gaussian}.

\begin{figure}[ht!]
\begin{tabular}{cc}
\begin{minipage}[c]{0.45\linewidth}
\includegraphics[width=\linewidth]{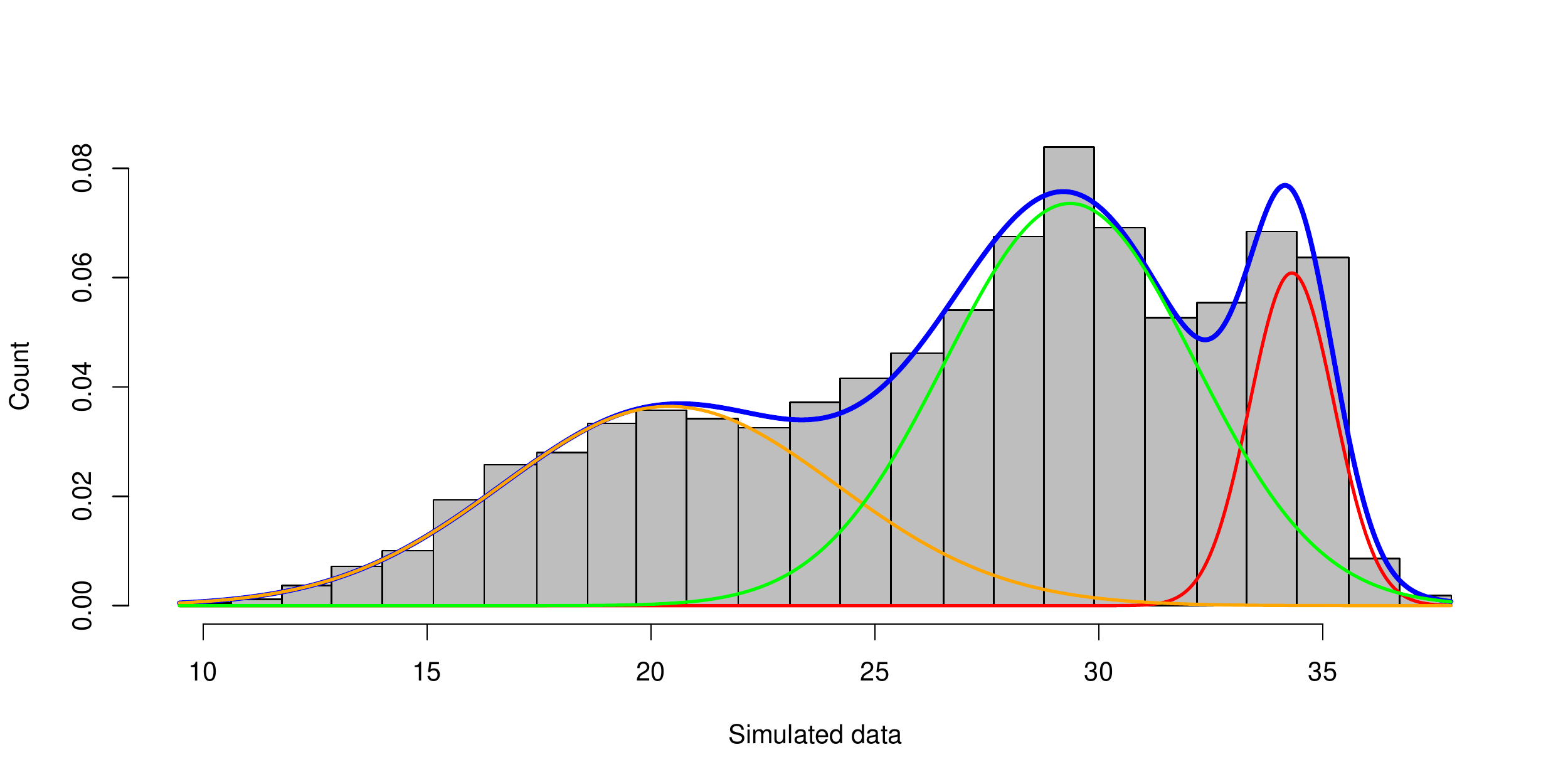}
\end{minipage}&
\begin{minipage}[c]{0.45\linewidth}
\includegraphics[width=\linewidth]{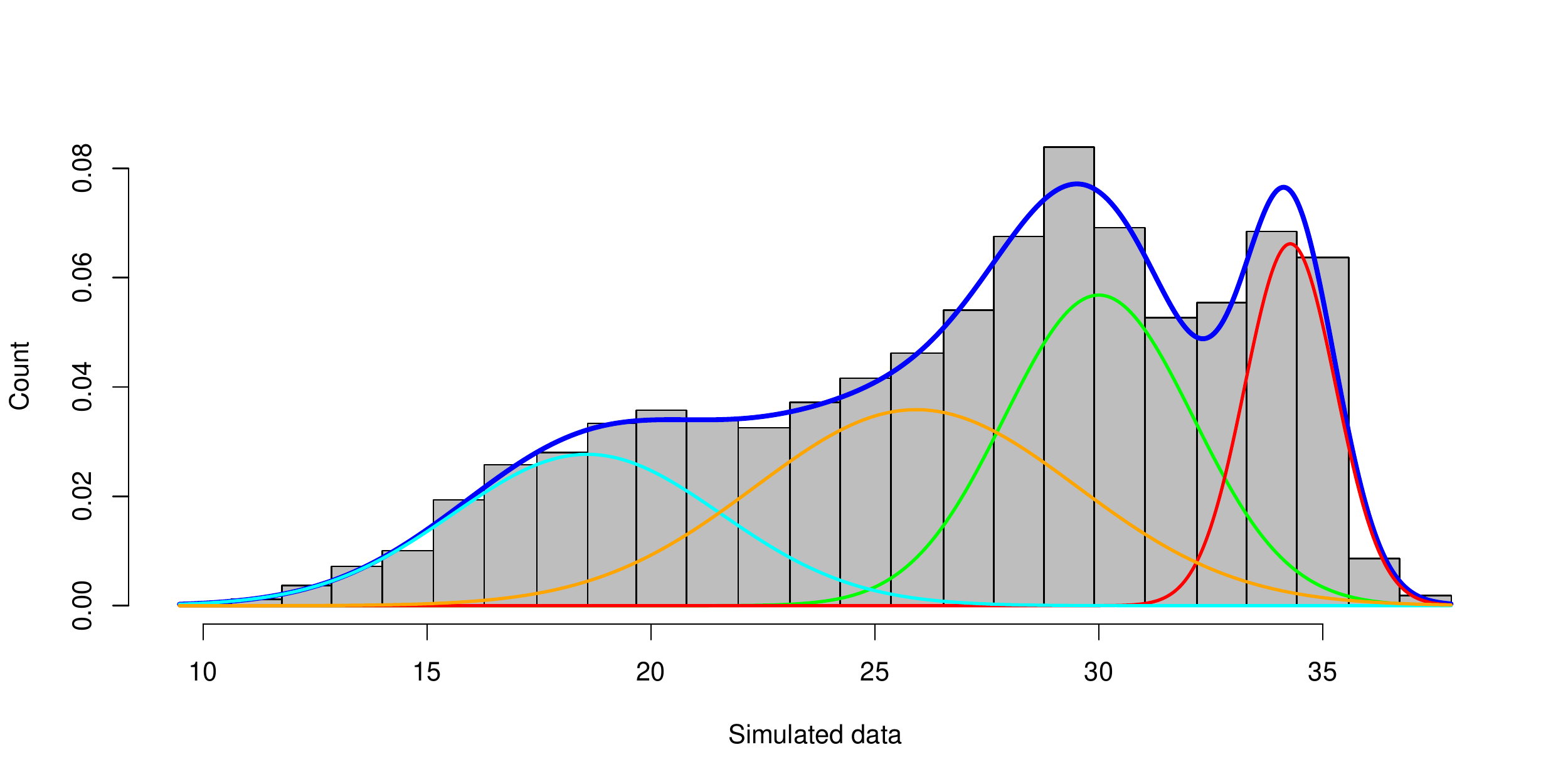}
\end{minipage}
\end{tabular}
\caption{\label{Fig:Histo:MM}Representation of the histogram with the densities estimated with 3 classes (on the left) and 4 classes (on the right): the color lines (other than blue) represent the density of each component and the blue line the density of the mixture.}
\end{figure}

As we do not have access to raw data, we performed simulations to  generate a reconstructed datasets with consistent histograms to Fig. 1 from \cite{Jones}, with randomized position of the points within each class. We applied the above procedure to 100 independently reconstructed data, in order to limit the influence of the random part. Among these 100 simulations, we obtain 95 times 3 clusters and 5 times 4 clusters. When there are 3 clusters, the estimation of the parameters is very stable (standard deviation less than $0.03$ for each) but there is a little more variability in the case of 4 clusters in particular for the two classes with the largest averages (but the standard deviation does not exceed $0.25$).

\subsection{Proof of the Theorem~\ref{lemma:censure}}
\label{sec:proofCensure}

To prove the lemma~\ref{lemma:censure}, we observe that for every $x\in\mathbb{R}$ we have the following decomposition of the density $f_X$ if $q>0$:

\begin{equation}\label{Eq:SI:Decomp}
f_X(x)=b\left(\mathbf{\eta}\right)\exp\left[\left<\mathbf{\eta},T(x)\right>\right],
\end{equation}
with $<\cdot,\cdot>$ is the scalar product,
$\mathbf{\eta}=\begin{pmatrix}
\frac{\mu}{\sigma^2},&-\frac{1}{2\sigma^2},&\ln q\\
\end{pmatrix}$ the natural parameters, $$T(x)=\begin{pmatrix}
x,&x^2,&\mathds{1}_{\{x>\dcens\}}\\
\end{pmatrix},$$
the sufficient statistics and 
$$b\left(\mathbf{\eta}\right)=\frac{1}{q+(1-q)F_{\mu,\sigma}(s)}\times\frac{1}{\sqrt{2\pi\sigma^2}}\exp\left(-\frac{\mu^2}{2\sigma^2}\right).$$ 

For the totally censured model, we have the same decomposition with the third parameters and taking $q=0$.
Thanks the decomposition in~\refn{Eq:SI:Decomp}, the (partially) censured model belongs to the family of exponential laws and the maximum likelihood estimators are strongly consistent and asymptotically normal.

\subsection{Simulations}\label{sec:SI:Simulations}

To study the quality of the estimators defined in Sec.~\ref{Sec:censored}, we simulated $10^4$ samples of size $n\in\{10^2,10^3,10^4,10^5\}$ of variables following the model $\mathcal{CN}_{\dcens}\left(0,1,p\right)$ with $\dcens\in\{-2,-1,0,1,2,3\}$ and $q\in\{0,0.1,0.5,0.9\}$. We provide boxplots estimations of the parameters in Fig.~\ref{Fig:boxplots} and a zoom on significant part in Fig.~\ref{Fig:boxplots:zoom}. Note that these parameters ($\mu = 0, |\dcens| \leq 3$) are very different from the ones expected for $C_t$ values, but the model can be straightforwardly adapted by an affine transformation to measured parameters of interest.

Observe from Fig.~\ref{Fig:boxplots} that the estimations are generally close to the parameters but we can sometimes have very large deviations. We find that the more $n$ increases, the better the estimator. The threshold seems to have a weak influence on the estimation of the partially censored model but, for the fully censored model, we see that the more $\dcens$ increases and the more the quality of the estimators increases; especially when $\dcens=-2$ which represents approximately the 2.3\% quantile. Note that we observe large deviations in the partially censored model when $\dcens$ is equal to~2; this may seem counter-intuitive since we have access to around 97.7\% of uncensored Gaussian information. However, this leaves few observations for the estimation of $p$ (which we observe on the graphs of the last line) and this weakens in this case the model because censorship no longer really has any reason to be. We therefore recommend using the model only when the number of observations after censorship is sufficient to estimate the parameter $p$.

\begin{figure}[t]
\begin{center}
\begin{tabular}{c||c|c|c}
Completely&\multicolumn{3}{|c}{Partially}\\
\cline{2-4}
$p=0$&$q=0.1$&$q=0.5$&$q=0.9$\\
\hline
&&&\\
\begin{minipage}[c]{0.225\linewidth}
\includegraphics[width=\linewidth]{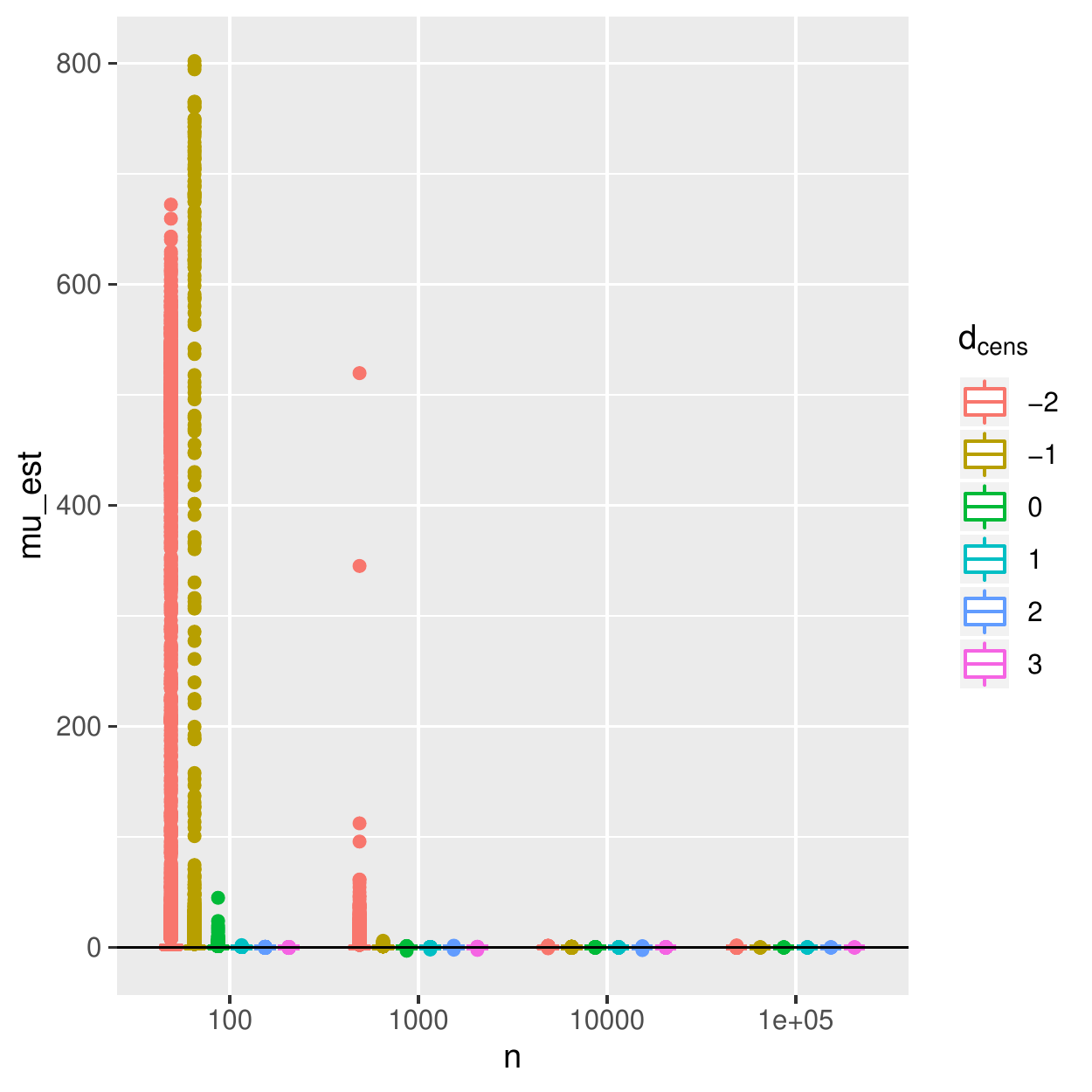}
\end{minipage}&
\begin{minipage}[c]{0.225\linewidth}
\includegraphics[width=\linewidth]{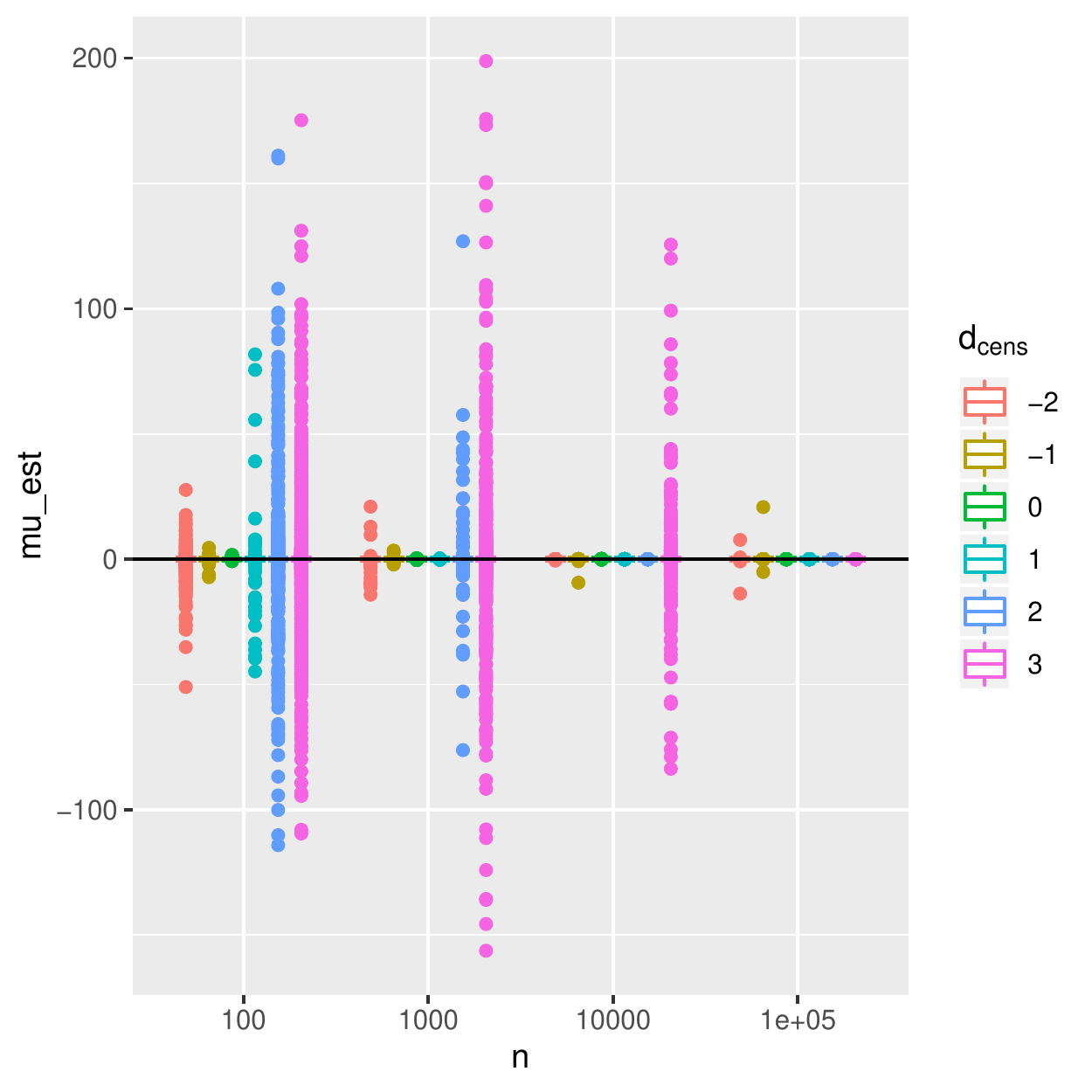}
\end{minipage}&
\begin{minipage}[c]{0.225\linewidth}
\includegraphics[width=\linewidth]{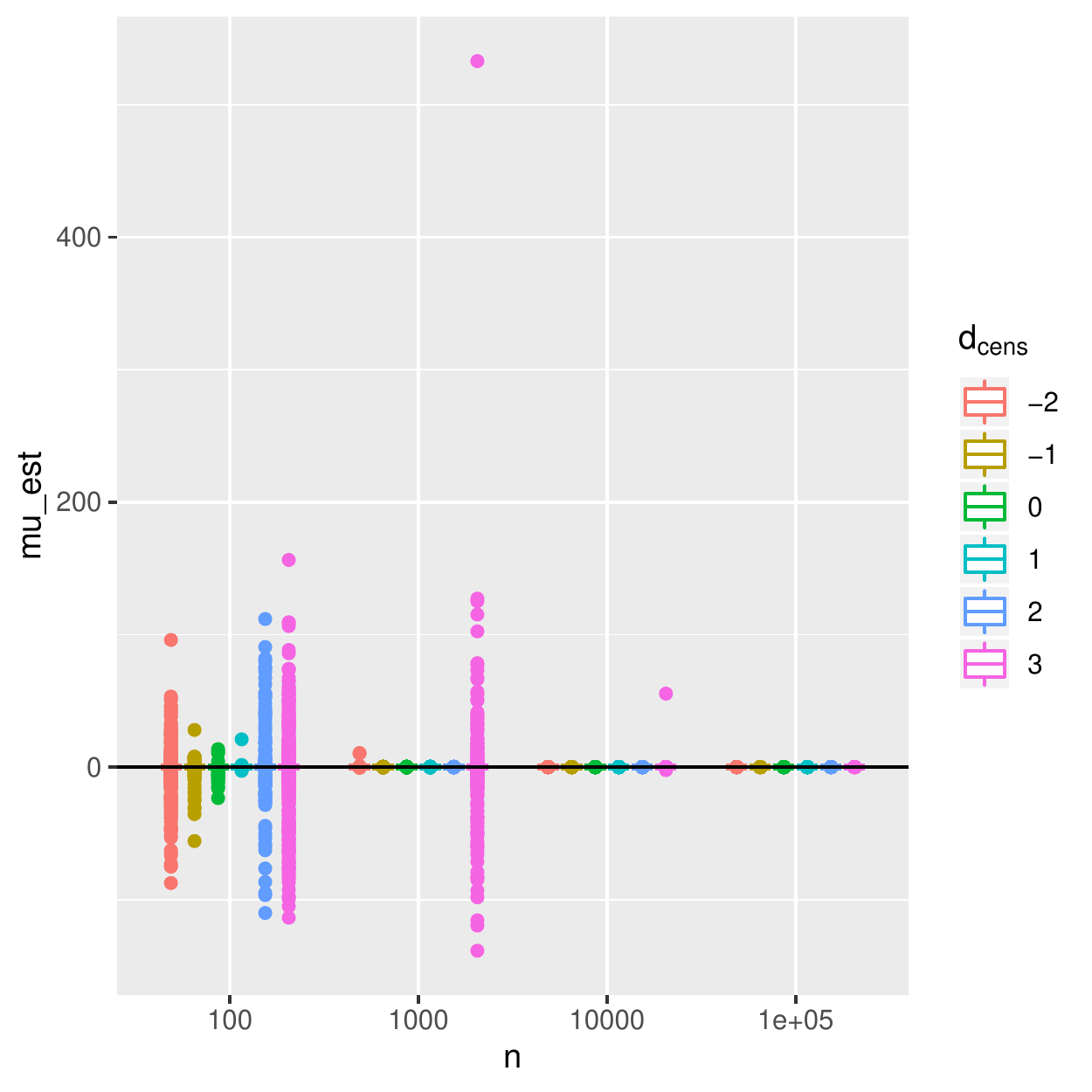}
\end{minipage}&
\begin{minipage}[c]{0.225\linewidth}
\includegraphics[width=\linewidth]{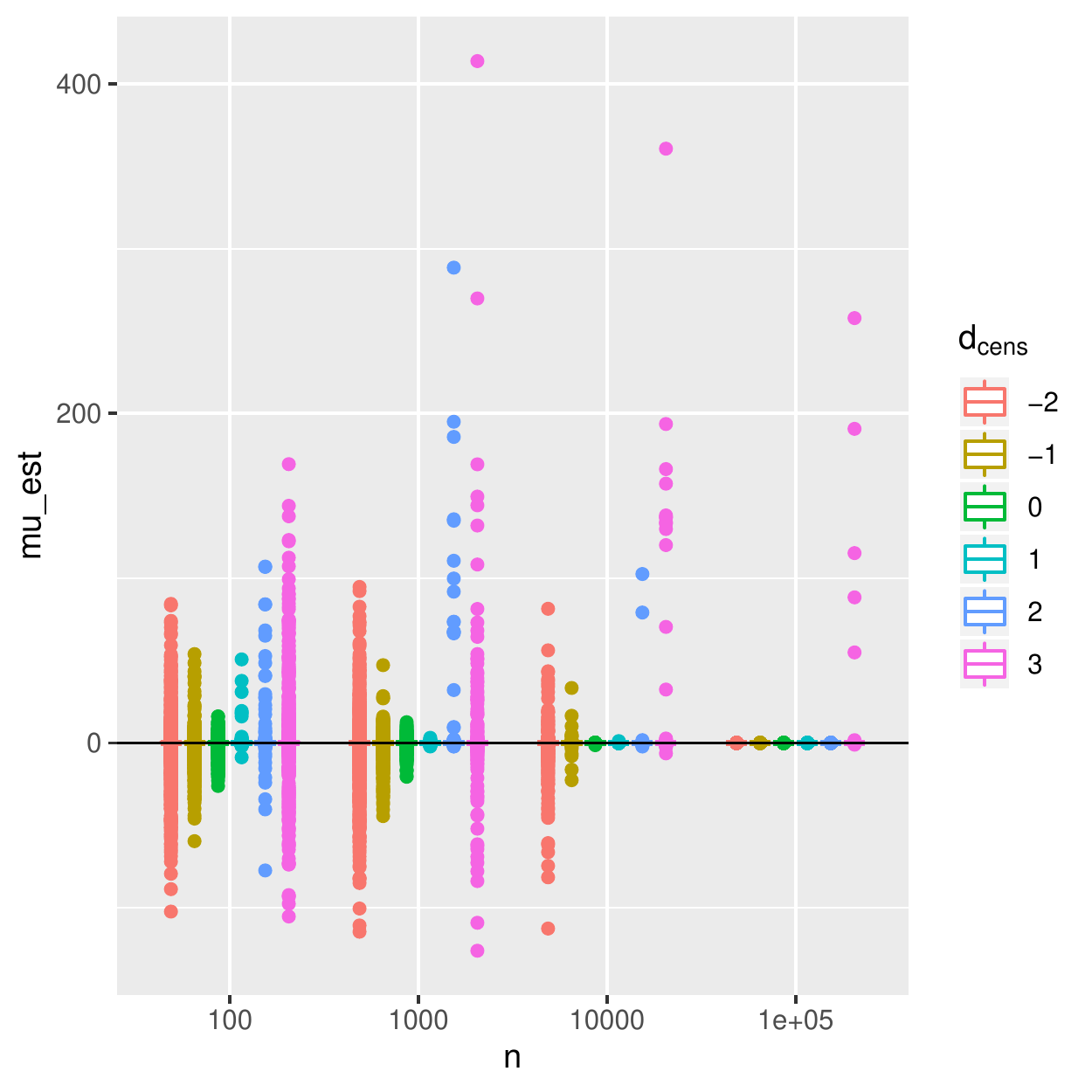}
\end{minipage}\\
&&&\\
\hline
&&&\\
\begin{minipage}[c]{0.225\linewidth}
\includegraphics[width=\linewidth]{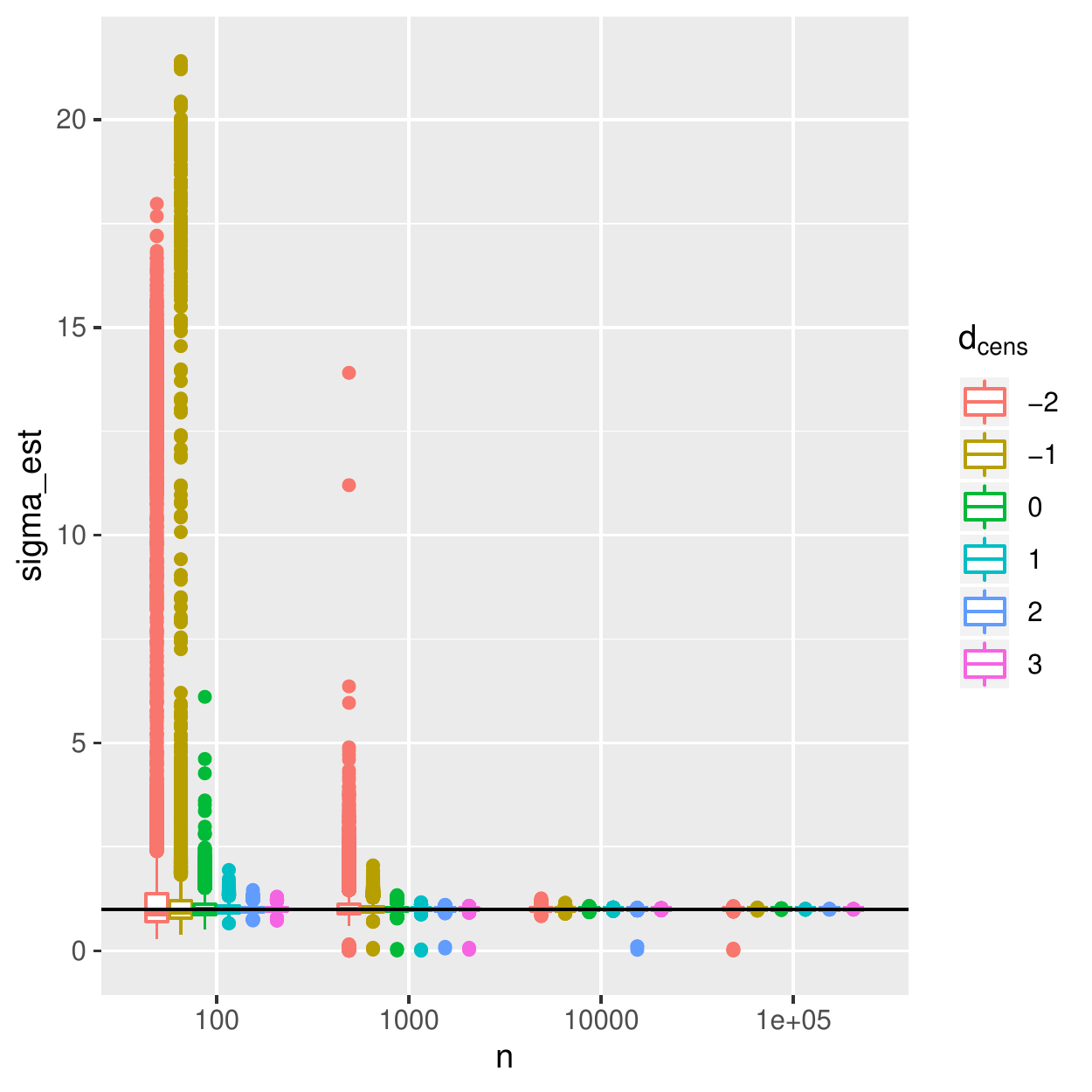}
\end{minipage}&
\begin{minipage}[c]{0.225\linewidth}
\includegraphics[width=\linewidth]{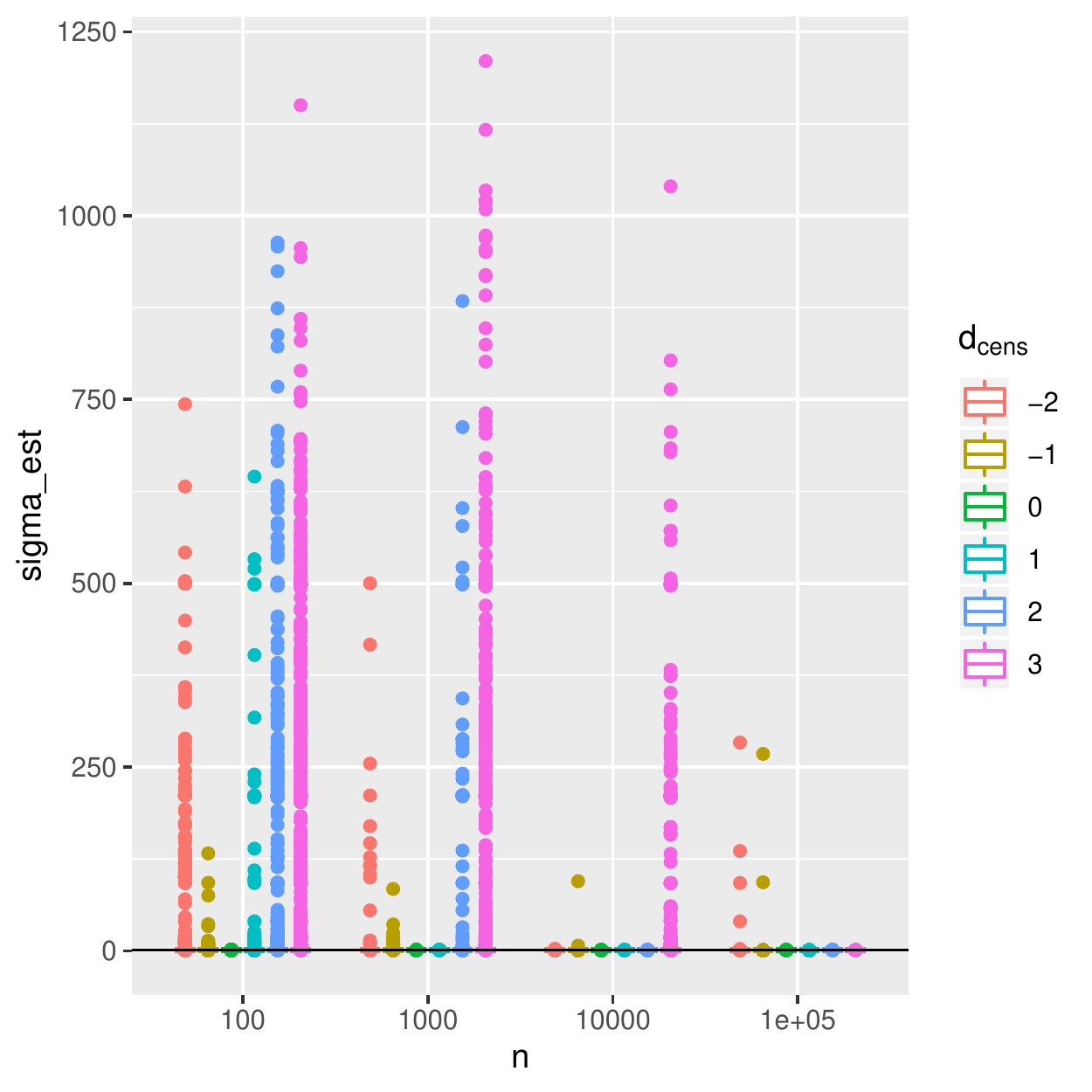}
\end{minipage}&
\begin{minipage}[c]{0.225\linewidth}
\includegraphics[width=\linewidth]{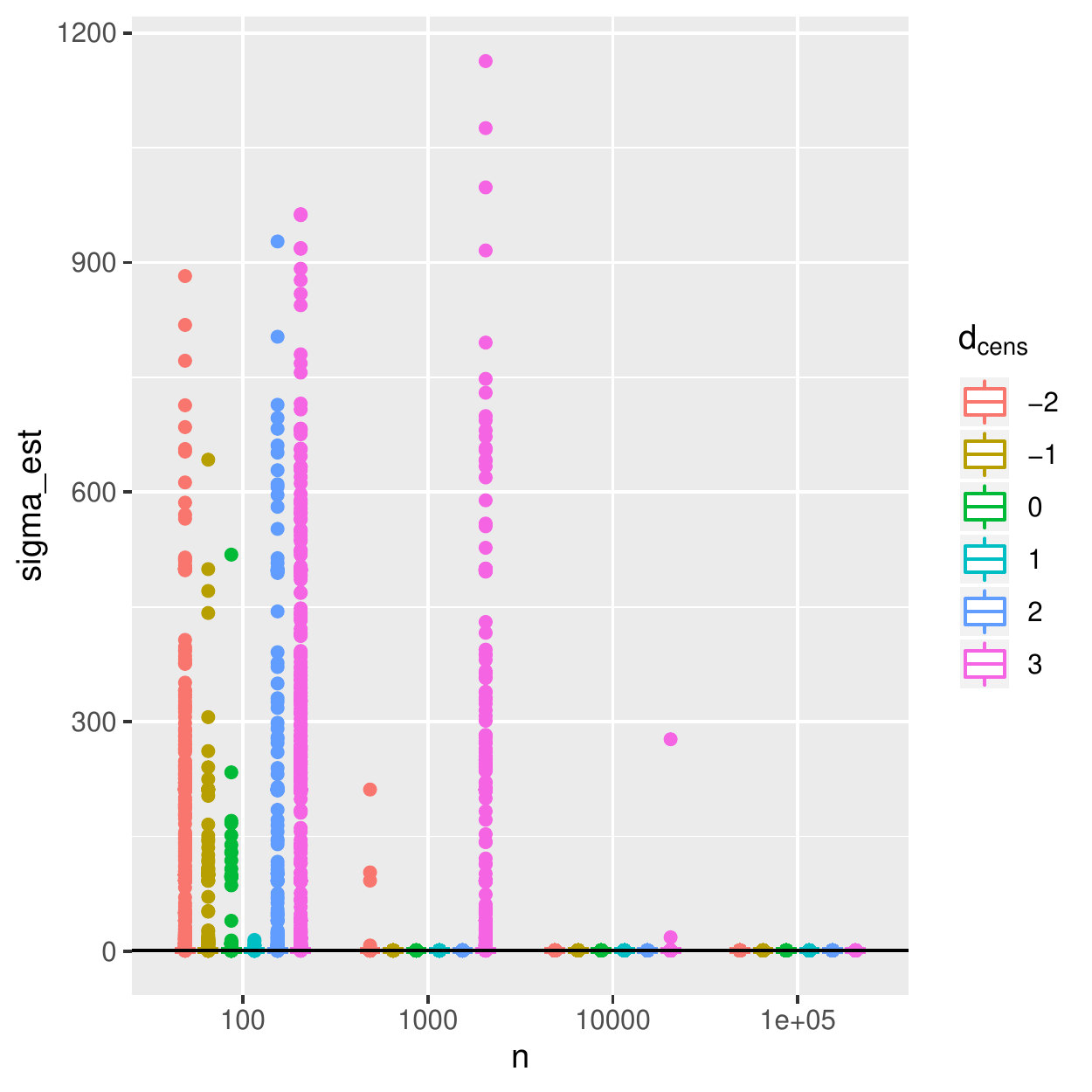}
\end{minipage}&
\begin{minipage}[c]{0.225\linewidth}
\includegraphics[width=\linewidth]{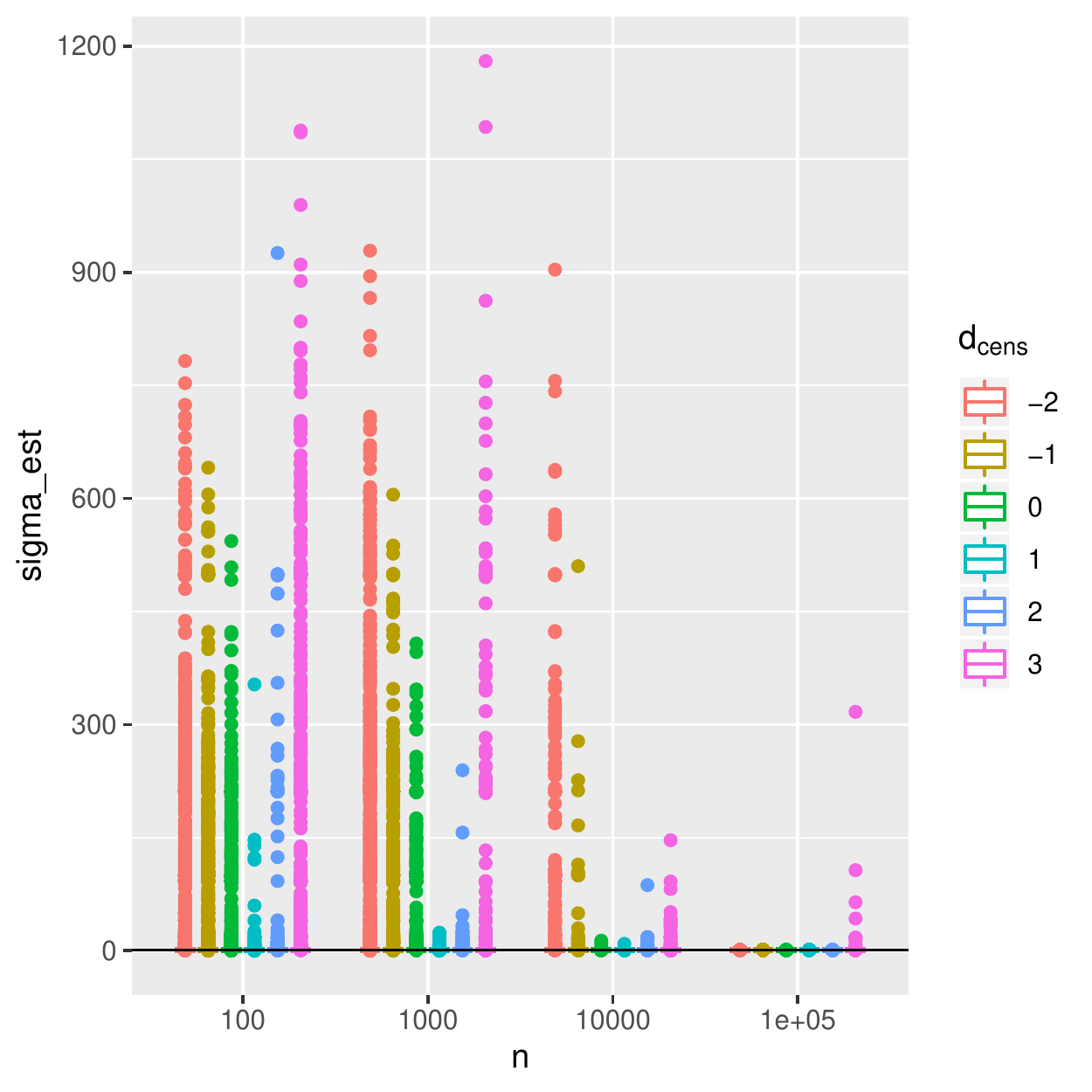}
\end{minipage}\\
&&&\\
\hline
&&&\\
&
\begin{minipage}[c]{0.225\linewidth}
\includegraphics[width=\linewidth]{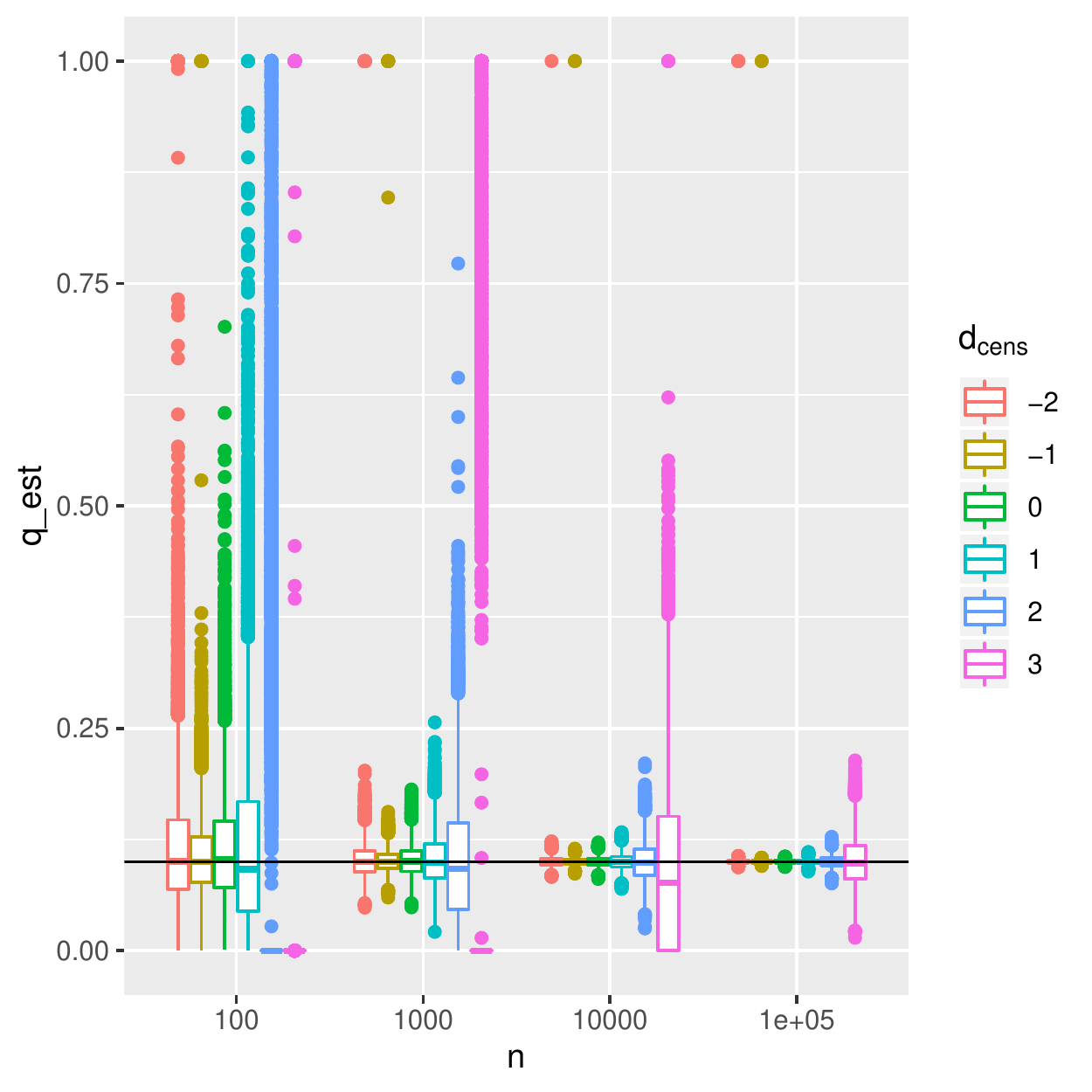}
\end{minipage}
&
\begin{minipage}[c]{0.225\linewidth}
\includegraphics[width=\linewidth]{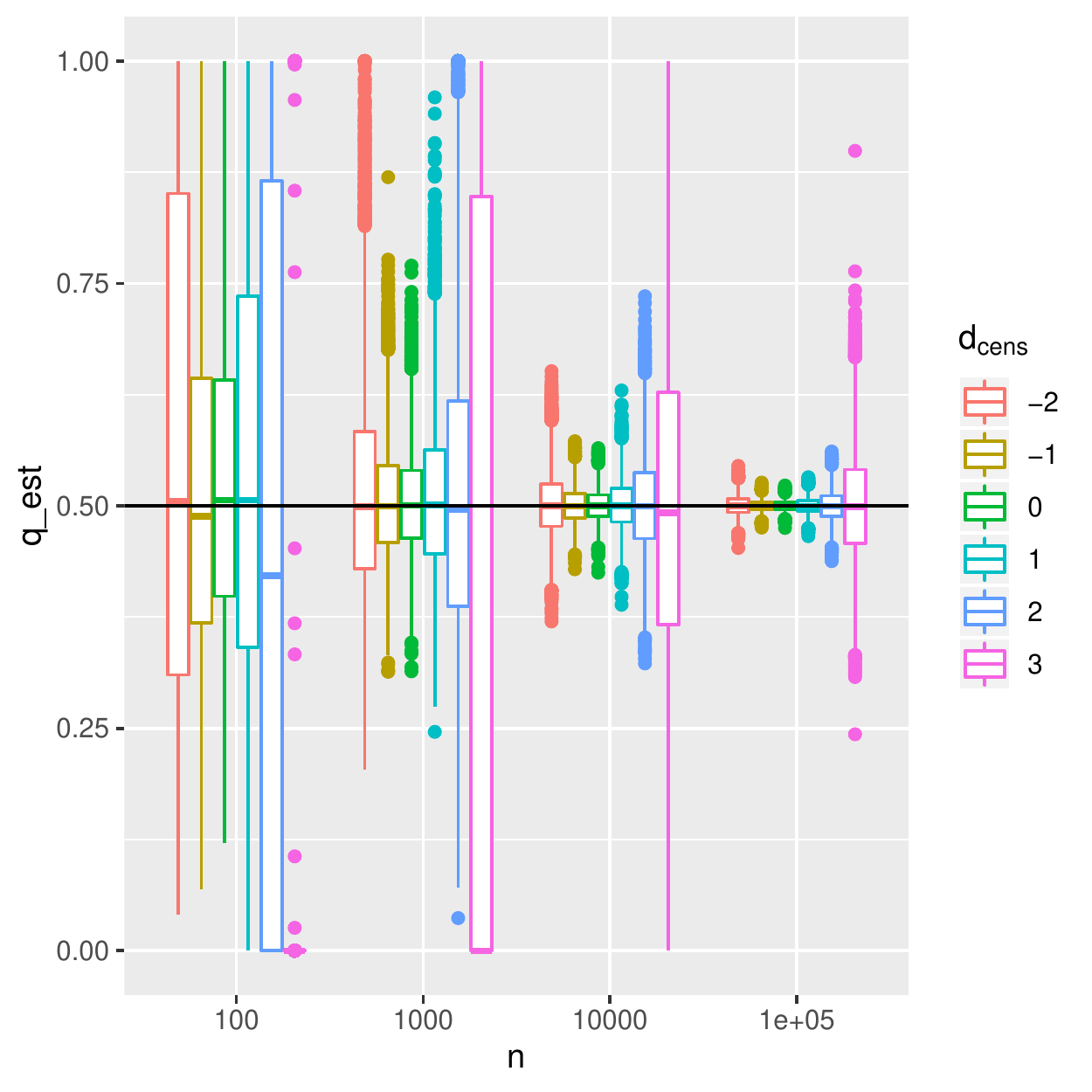}
\end{minipage}
&
\begin{minipage}[c]{0.225\linewidth}
\includegraphics[width=\linewidth]{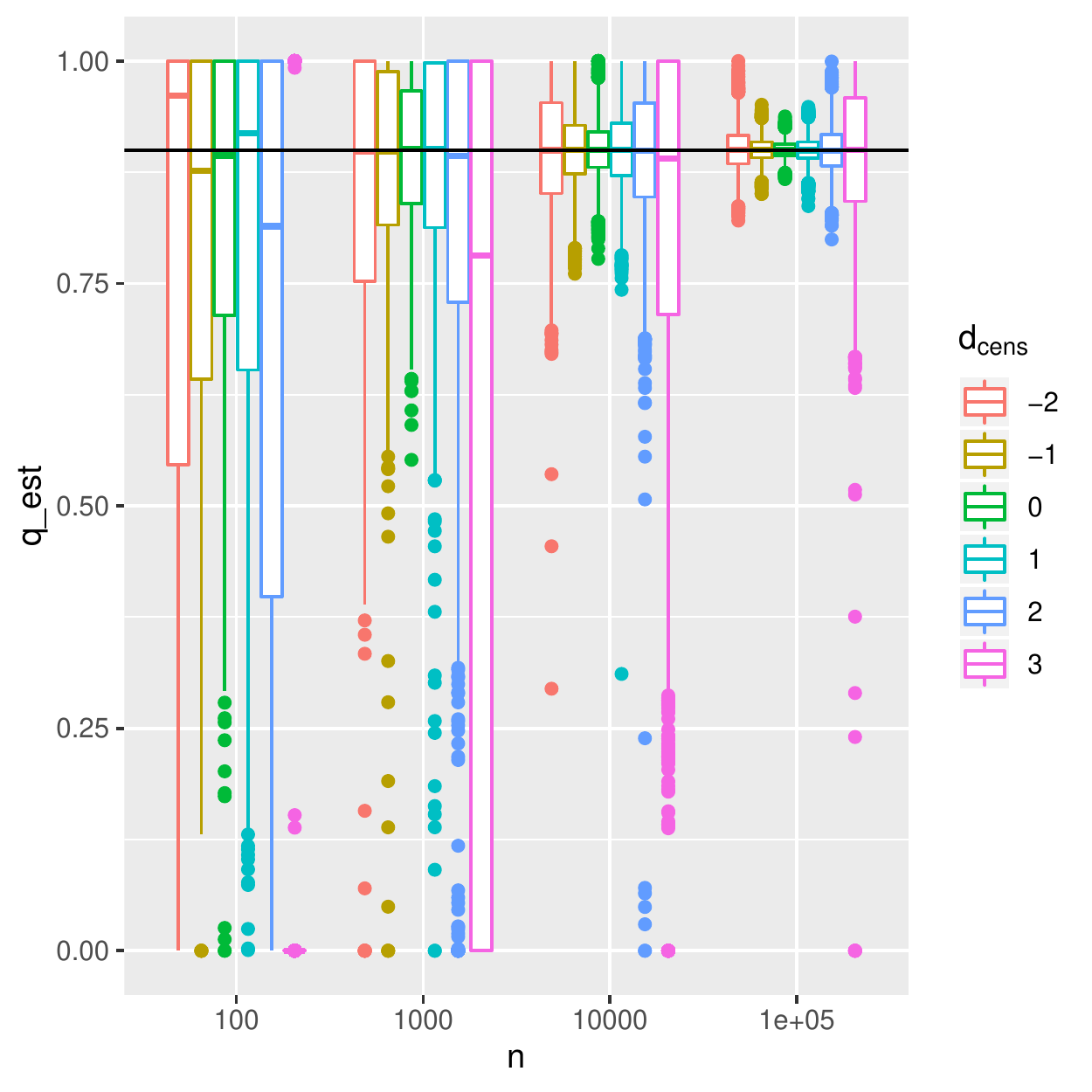}
\end{minipage}\\
\end{tabular}
\caption{\label{Fig:boxplots}Boxplots of the estimations of $\mu$ (first row), $\sigma$ (second row) and $p$ (last row~; only for partially censored model) in function of model (columns), the size $n$ of sample (x-axis) and the value of the threshold $s$ (color). The true value is symbolised by the horizontal black line.}
\end{center}
\end{figure}

\begin{figure}[t]
\begin{center}
\begin{tabular}{c||c|c|c}
Completely&\multicolumn{3}{|c}{Partially}\\
\cline{2-4}
$q=0$&$q=0.1$&$q=0.5$&$q=0.9$\\
\hline
&&&\\
\begin{minipage}[c]{0.225\linewidth}
\includegraphics[width=\linewidth]{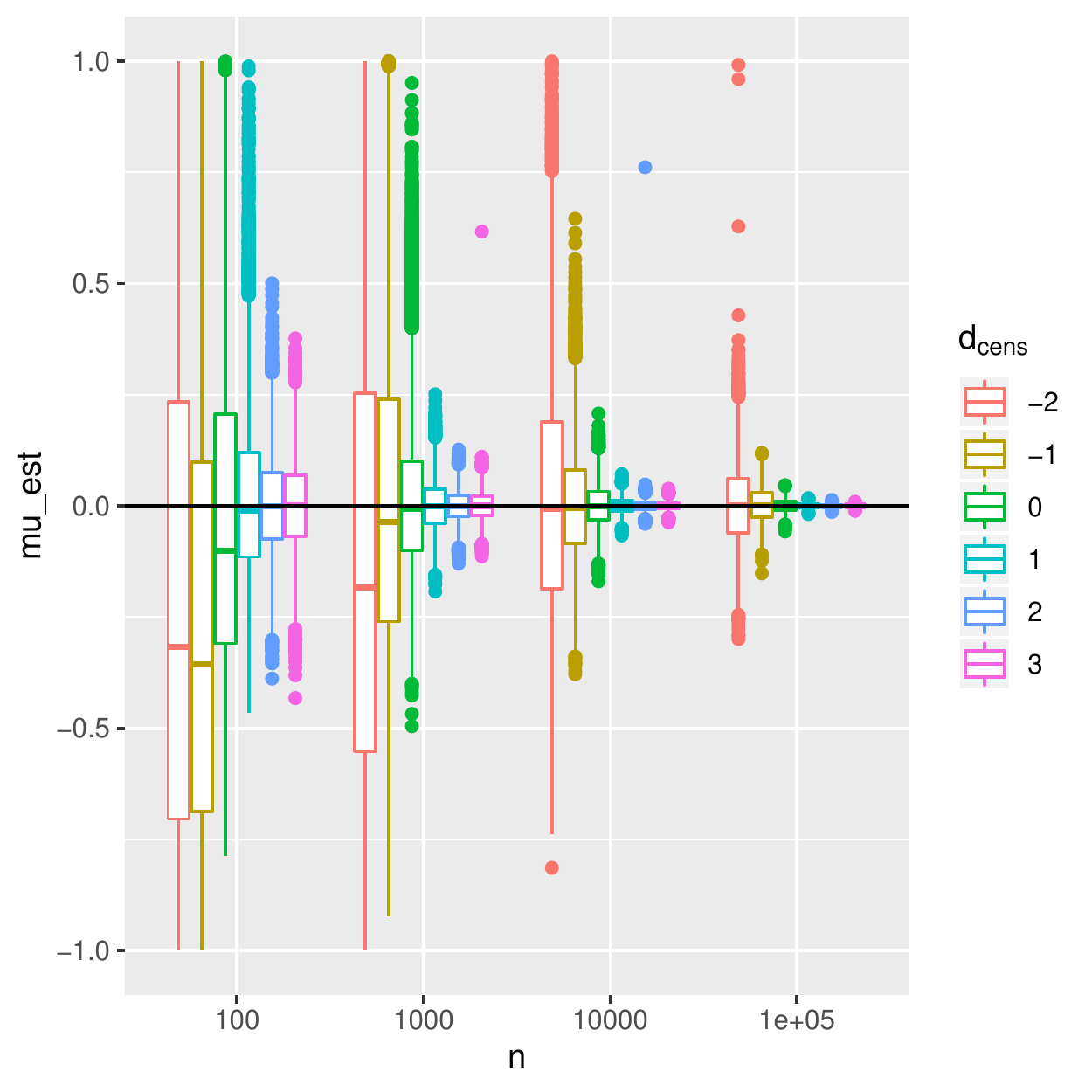}
\end{minipage}&
\begin{minipage}[c]{0.225\linewidth}
\includegraphics[width=\linewidth]{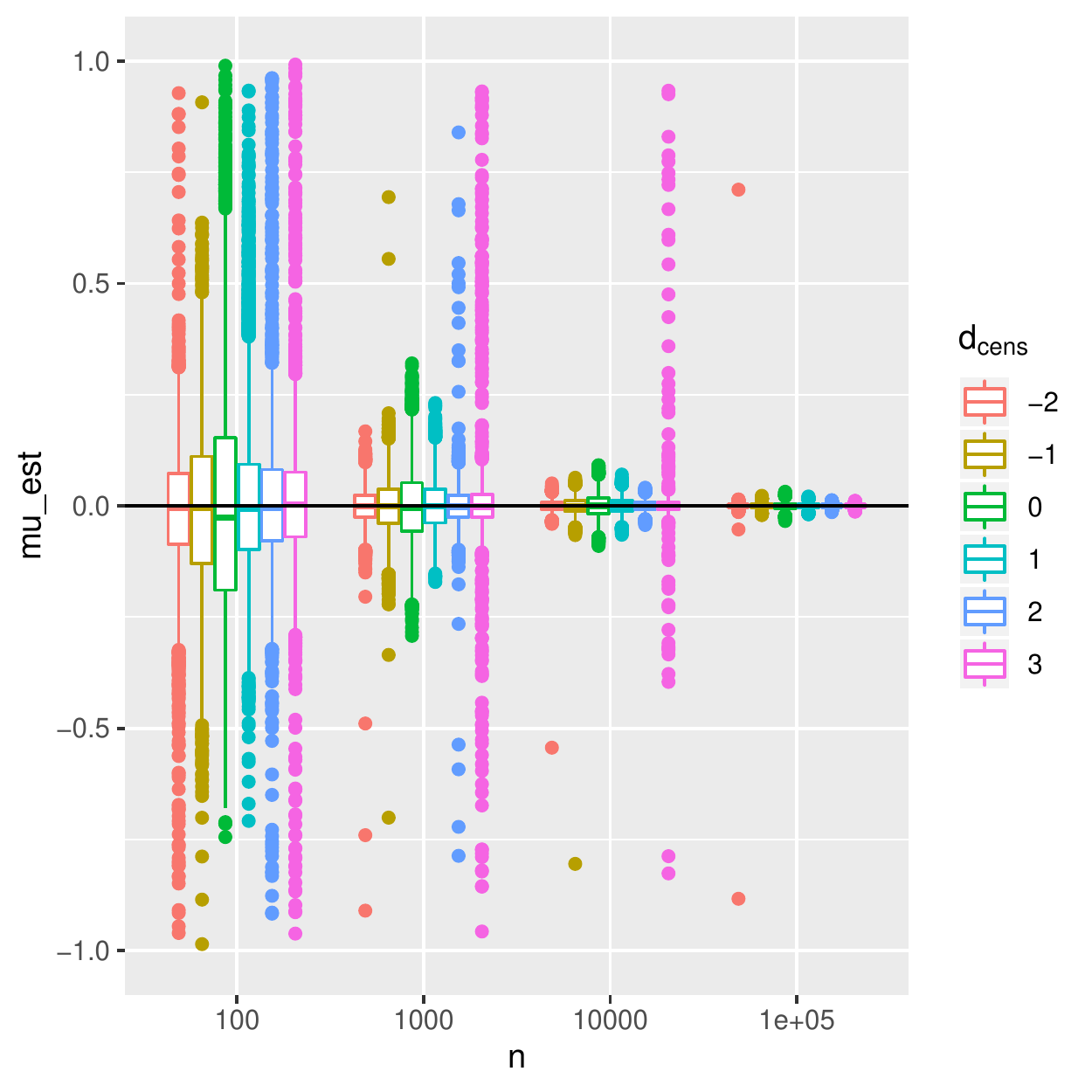}
\end{minipage}&
\begin{minipage}[c]{0.225\linewidth}
\includegraphics[width=\linewidth]{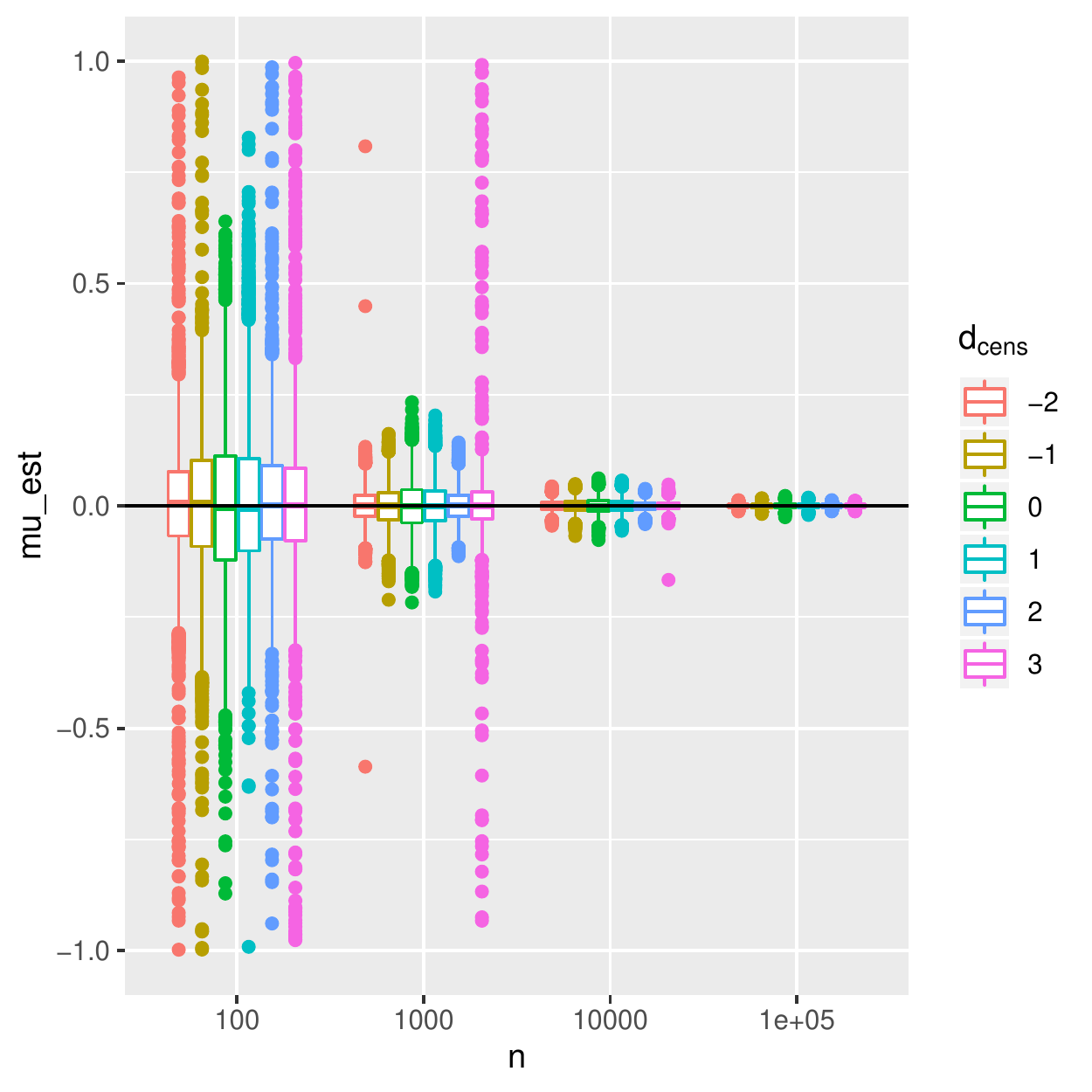}
\end{minipage}&
\begin{minipage}[c]{0.225\linewidth}
\includegraphics[width=\linewidth]{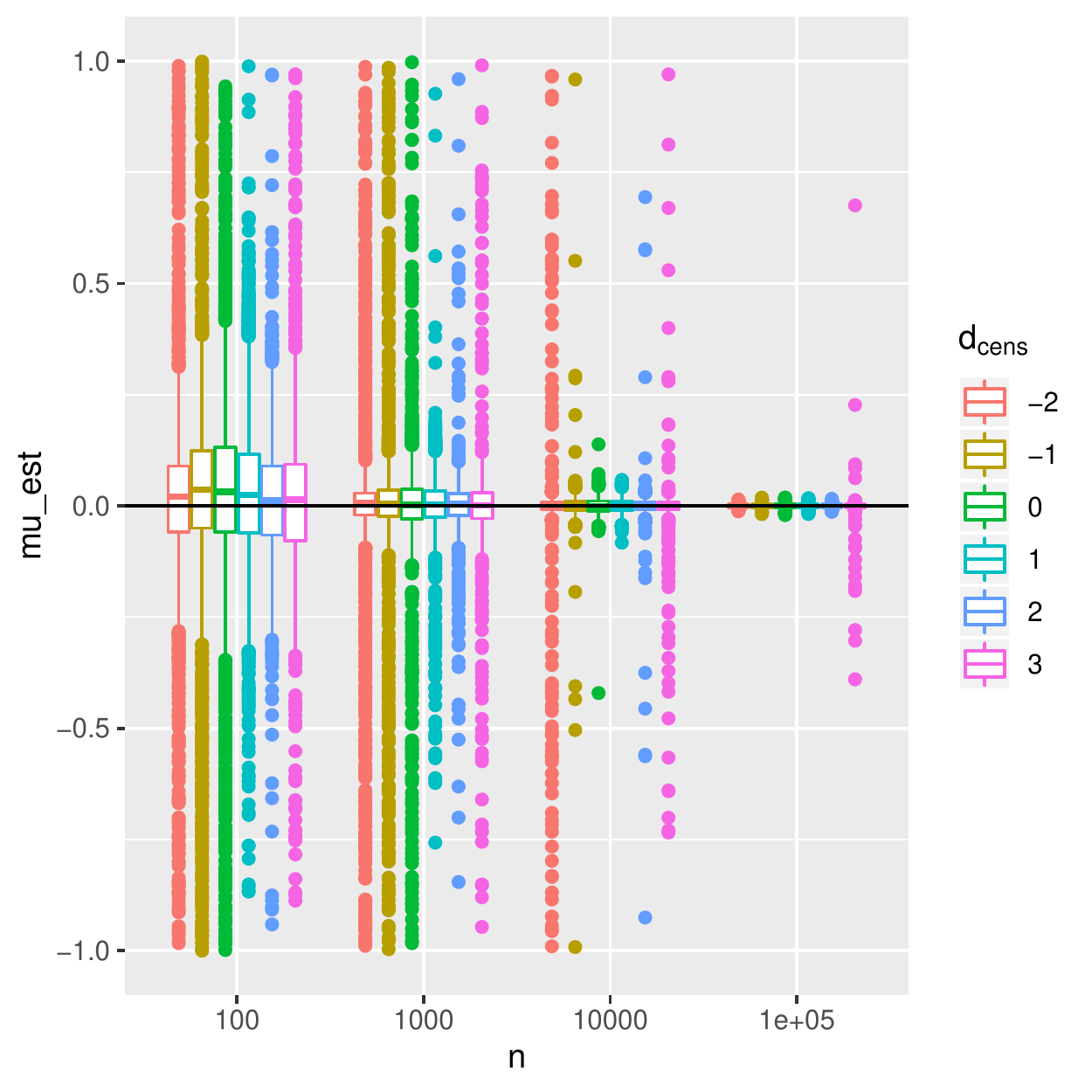}
\end{minipage}\\
&&&\\
\hline
&&&\\
\begin{minipage}[c]{0.225\linewidth}
\includegraphics[width=\linewidth]{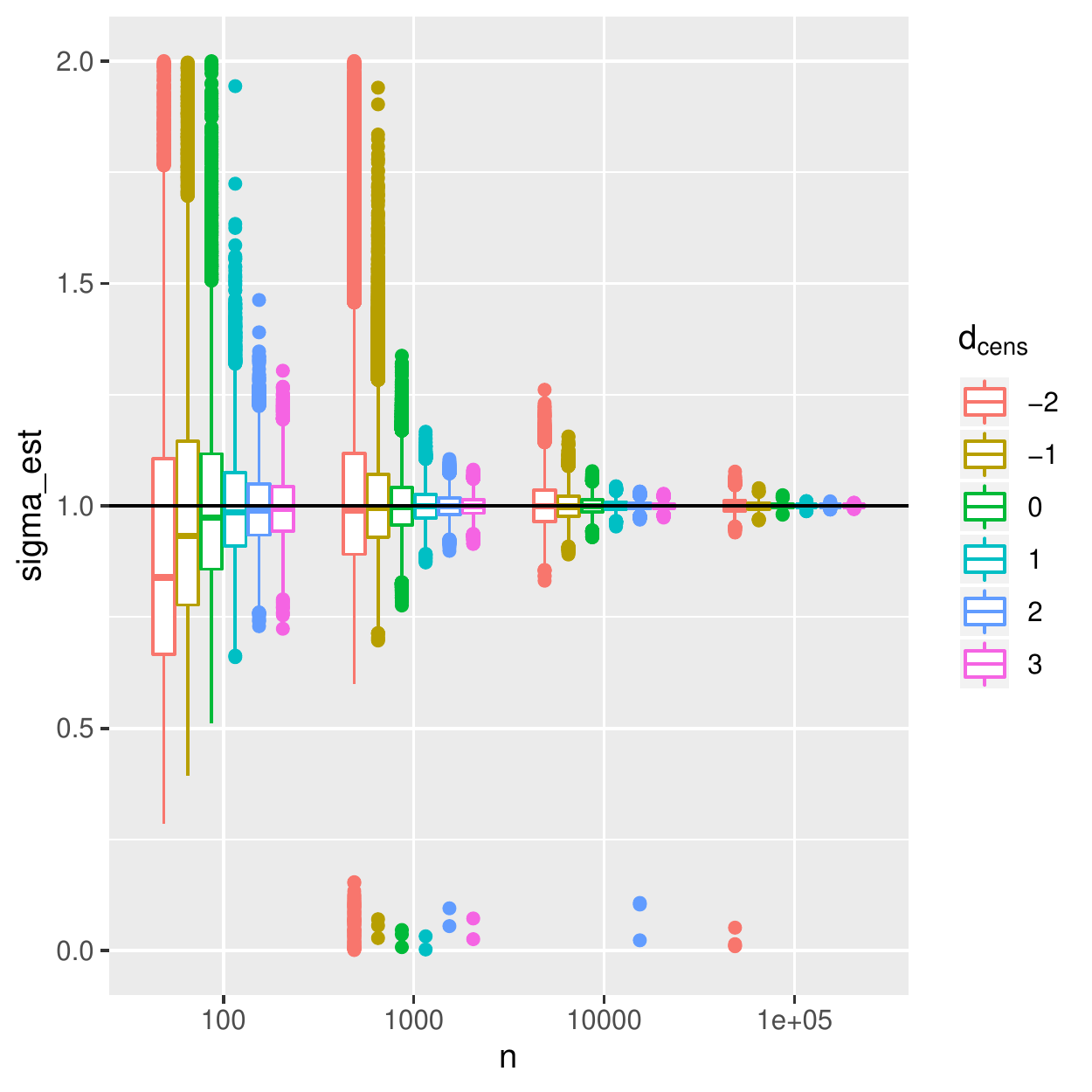}
\end{minipage}&
\begin{minipage}[c]{0.225\linewidth}
\includegraphics[width=\linewidth]{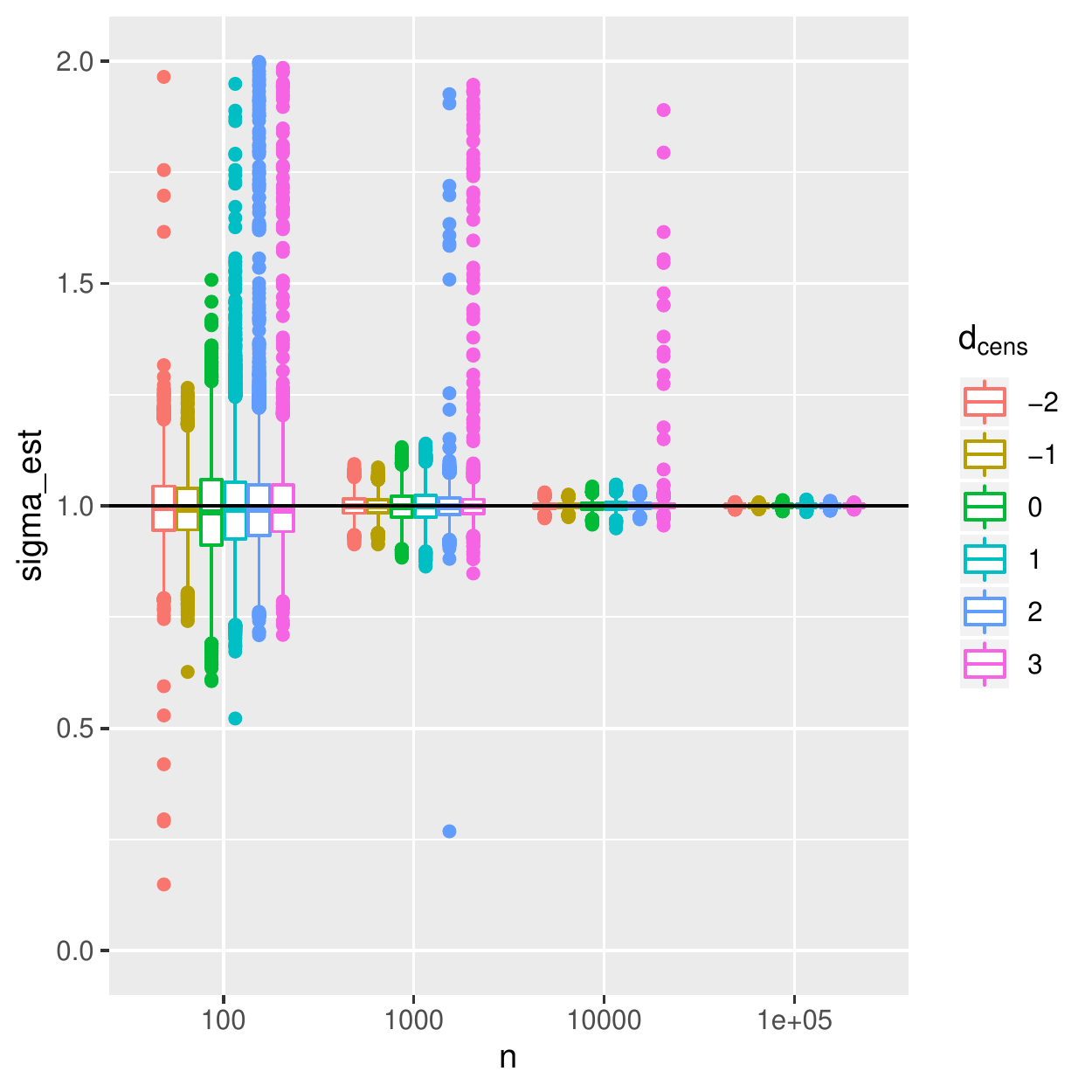}
\end{minipage}&
\begin{minipage}[c]{0.225\linewidth}
\includegraphics[width=\linewidth]{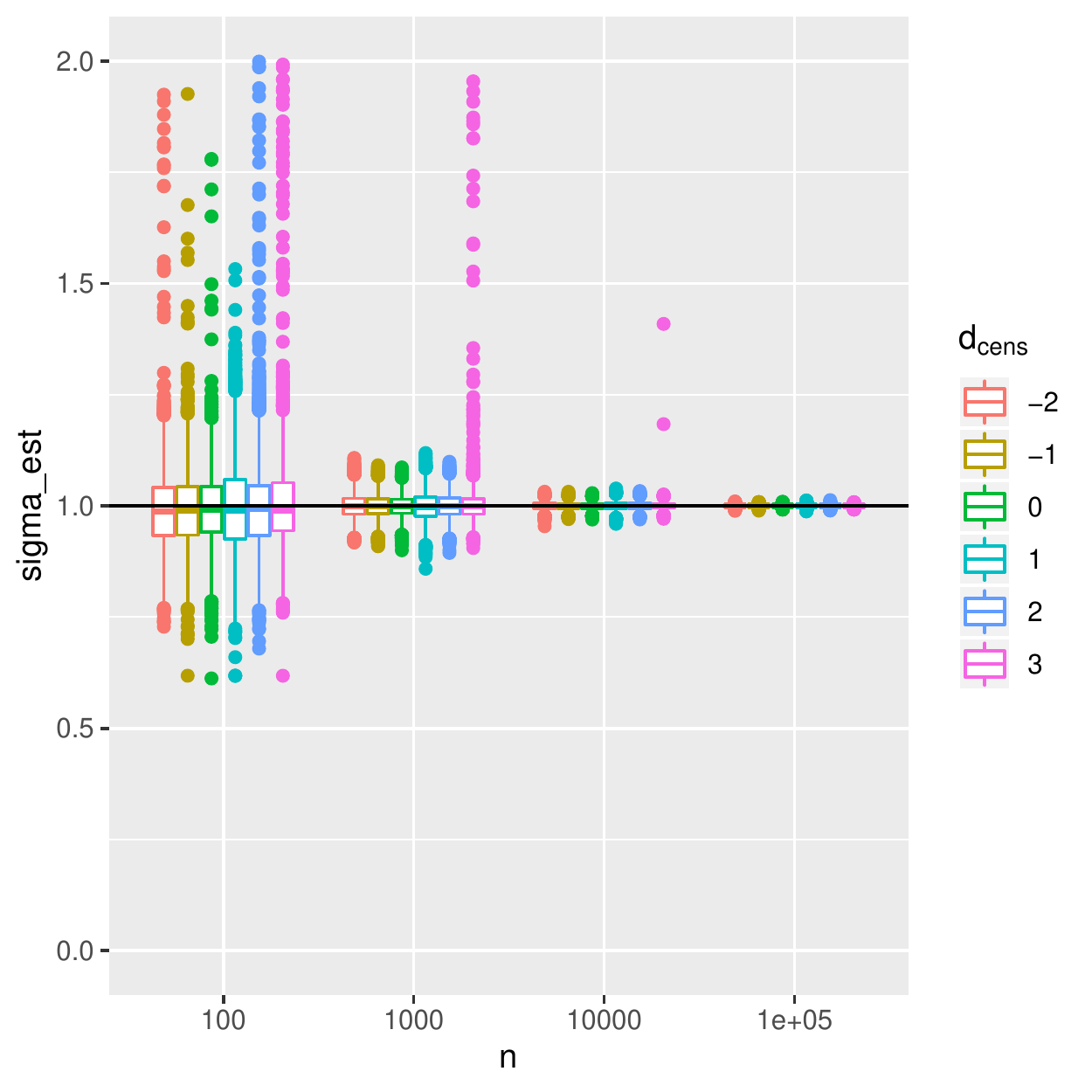}
\end{minipage}&
\begin{minipage}[c]{0.225\linewidth}
\includegraphics[width=\linewidth]{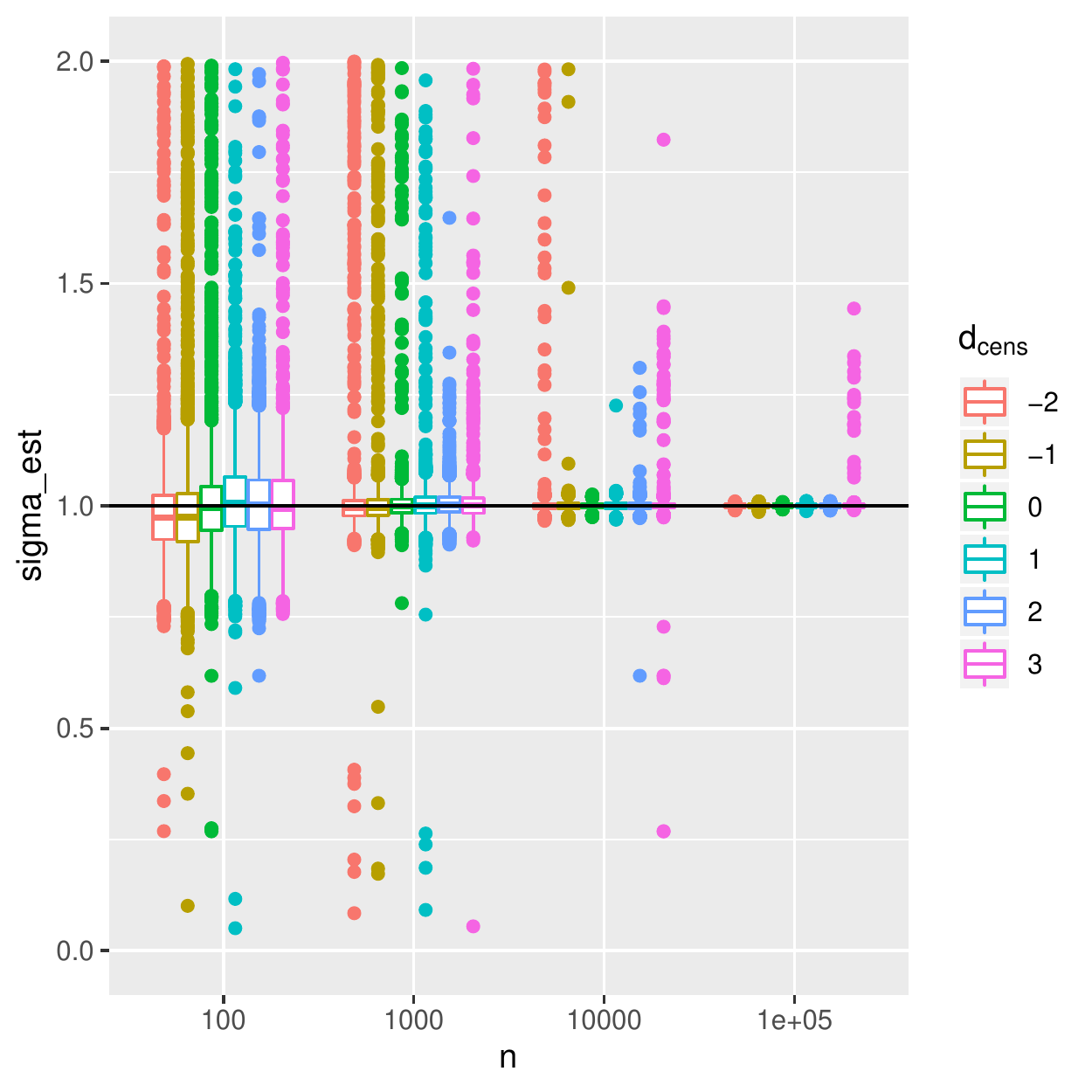}
\end{minipage}\\
&&&\\
\hline
&&&\\
&
\begin{minipage}[c]{0.225\linewidth}
\includegraphics[width=\linewidth]{boxplot_Inv_prop_p_1}
\end{minipage}
&
\begin{minipage}[c]{0.225\linewidth}
\includegraphics[width=\linewidth]{boxplot_Inv_prop_p_5}
\end{minipage}
&
\begin{minipage}[c]{0.225\linewidth}
\includegraphics[width=\linewidth]{boxplot_Inv_prop_p_9}
\end{minipage}\\
\end{tabular}
\caption{\label{Fig:boxplots:zoom}Zoom on boxplots of the estimations of $\mu$ (first row), $\sigma$ (second row) and $q$ (last row~; only for partially censored model) in function of model (columns), the size $n$ of sample (x-axis) and the value of the threshold $s$ (color). The true value is symbolised by the horizontal black line.}
\end{center}
\end{figure}

\subsection{Censored mixture model}
\label{sec:stats}

In this section, we present the complementary graphs of Sec.~\ref{sec:statistical}. The statistical model presented here has the following density defined for all $x\in\mathbb{R}$ by:

\begin{equation}\label{Eq:SI:Dens}
f(x)=\sum_{k=1}^3\pi_k \frac{f_{\mu_k,\sigma_k}(x)}{q_k+(1-q_k)F_{\mu_k,\sigma_k}(\dcens)}\left[1+(q_k-1)\mathds{1}_{\{x> \dcens\}}\right].
\end{equation}

The completely censured mixture model has the same density than the \refn{Eq:SI:Dens} with $q_k=0$.

With the model in~\refn{Eq:SI:Dens}, we can estimate the theoretical false negative rate by the following formula:

\begin{equation}\label{Eq:SI:False}
\mathbb{P}\left(\text{false negative}\right)=\sum_{k=1}^3\pi_k\left[1-F_{\mu_k,\sigma_k}(\dcens)\right]\left(1-q_k\right).
\end{equation}

\begin {table}[ht!]
{
\begin{center}
\caption {Estimated parameters for the naive Gaussian mixture fit and the censored Gaussian mixture fits defined in \refn{Eq:SI:Dens}, for the datasets available in \cite{Jones} and \cite{Cabrera2020}. Note the consistency of the estimations, in particular in the partially and completely censored models.} \label{tab:MM_Gaussian}
\cite{Jones}
\begin{tabular}{|c||c||c|c|c||c|c|c||c|c|c|}
\hline
Model&${q_i}_{i=1..3}$&$\mu_1$ & $\sigma_1$ & $\pi_1$ & $\mu_2$ & $\sigma_2$  & $\pi_2$ & $\mu_3$ & $\sigma_3$ & $\pi_3$   \\
\hline
Naive&&20.41 & 3.74 & 0.34 &29.43 & 2.81 & 0.52 & 34.32 & 0.89 & 0.14\\
\hline
Partially &0.2&20.14 & 3.60 & 0.32 & 29.35& 2.96 & 0.53 & 34.78  & 1.32  & 0.14\\
\hline
Completely&&20.13 & 3.60 & 0.33 & 29.41 & 3.02 & 0.54 & 34.81 & 1.31 & 0.13\\
\hline
\end{tabular}
\vspace{.3cm}

\cite{Cabrera2020}
\begin{tabular}{|c||c||c|c|c||c|c|c||c|c|c|}
\hline
Model&${q_i}_{i=1..3}$&$\mu_1$ & $\sigma_1$ & $\pi_1$ & $\mu_2$ & $\sigma_2$  & $\pi_2$ & $\mu_3$ & $\sigma_3$ & $\pi_3$   \\
\hline
Naive&&19.75 & 2.05 & 0.20 &25.61 & 2.99 & 0.39 & 34.28 & 2.36 & 0.40\\
\hline
Partially &0.4&20.16 & 2.19 & 0.26 & 26.03& 2.58 & 0.43 & 34.54  & 2.66  & 0.41\\
\hline
Completely&&20.55 & 3.45 & 0.31 & 26.33 & 2.11 & 0.24 & 34.41 & 2.98 & 0.43\\
\hline
\end{tabular}
\end{center}
}
\end {table}

\begin{figure}[!h]
\begin{tabular}{ccc}
\begin{minipage}[c]{0.3\linewidth}
\includegraphics[width=\linewidth]{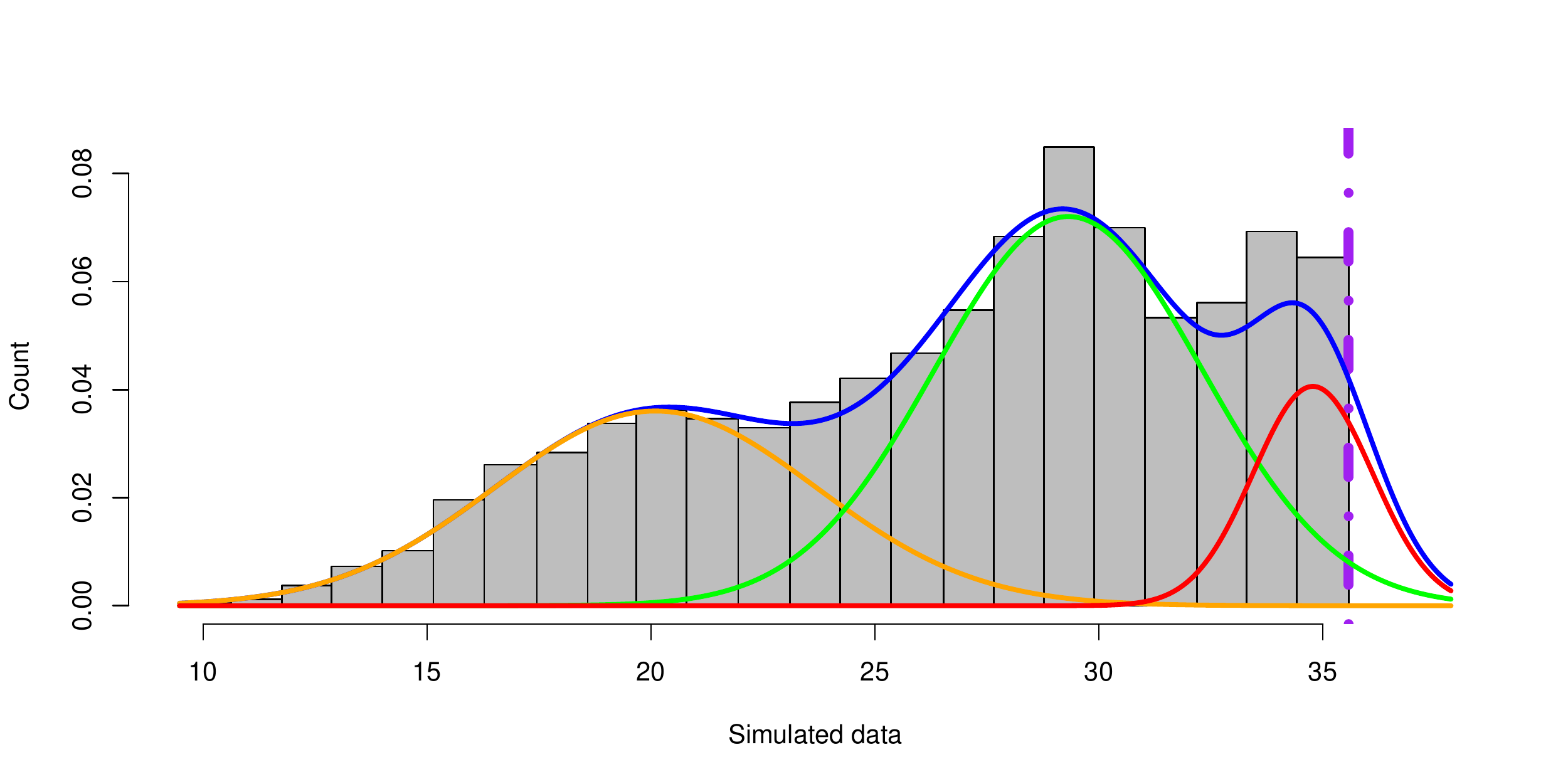}
\end{minipage}&
\begin{minipage}[c]{0.3\linewidth}
\includegraphics[width=\linewidth]{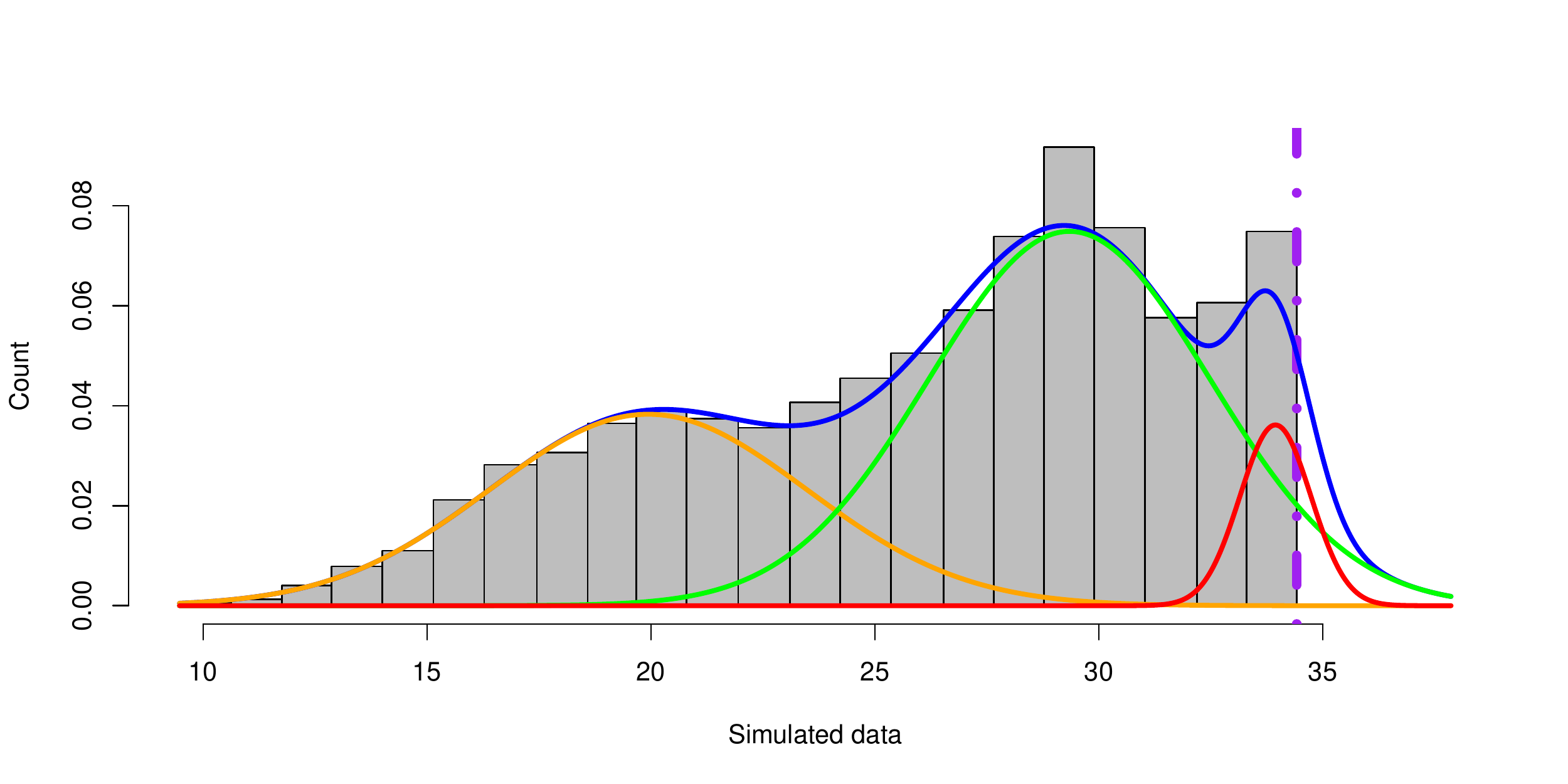}
\end{minipage}&
\begin{minipage}[c]{0.3\linewidth}
\includegraphics[width=\linewidth]{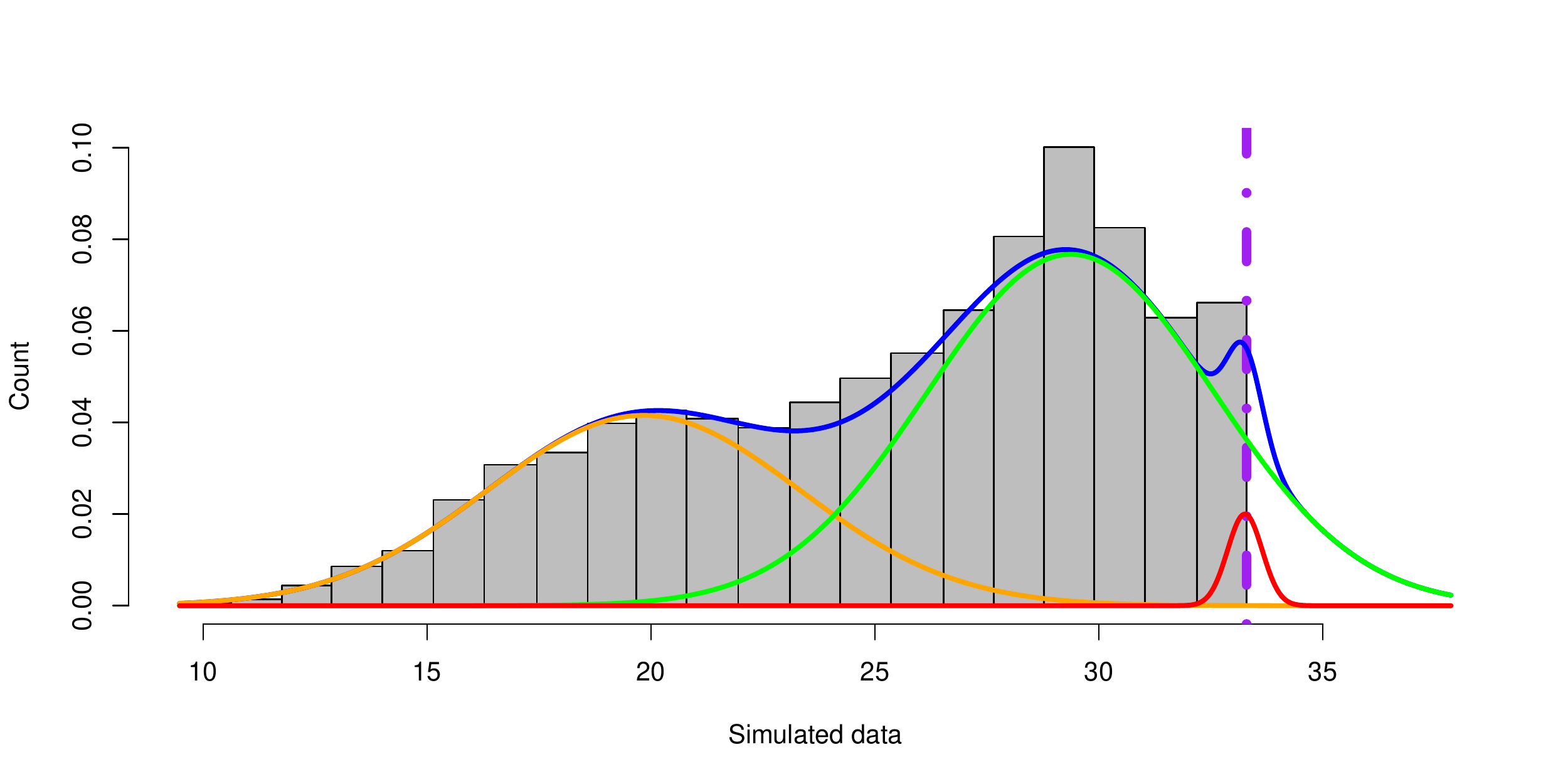}
\end{minipage}
\end{tabular}
\caption{\label{Fig:Histo:MM_Censure:long}Density of the fits of the censored model with three components (obtained when erasing data to the right of the threshold) with a threshold at 35.6 (left), 34.4 (middle) and 33.2 (right): the orange, green and red lines represent the density of each component and the blue line the density of the mixture. The histogram correspond to the one presented in \cite{Jones}.}
\end{figure}

\begin {table}[ht!]
{
\begin{center}
\caption {Estimated parameters for the censored Gaussian mixture fit define in \refn{Eq:SI:Dens} for different values of the threshold~$\dcens$, applied to reconstructed data data with same distribution as in~\cite{Jones} erased above $\dcens$.} \label{tab:MM_Gaussian:Comp} 
\begin{tabular}{|c||c|c|c||c|c|c||c|c|c|}
\hline
$\dcens$&$\mu_1$ & $\sigma_1$ & $\pi_1$ & $\mu_2$ & $\sigma_2$  & $\pi_2$ &  $\mu_3$ & $\sigma_3$ & $\pi_3$   \\
\hline
35.6&20.13 & 3.60 & 0.33 & 29.41 & 3.02 & 0.54 & 34.81 & 1.31 & 0.13\\
\hline
34.4&20.13 & 3.61 & 0.35 & 29.35 & 2.99 & 0.57 & 34.21 & 1.03 & 0.08 \\
\hline
33.2&19.97 & 3.56 & 0.03 & 29.40 & 3.14 & 0.59 & 33.21 & 1.16 & 0.48\\
\hline
\end{tabular}
\end{center}
}
\end {table}

\clearpage

\section{Estimation of the false-negative risk in the presence of multiple positive individuals in the pool}
\label{si:multiplepositive}

We treat here the case of a pool of $N$ samples that contains $k>1$ positive individuals. It is particularly relevant to consider these case when pooling correlated samples, (such has individuals living in the same household, or workers sharing the same area). Indeed,  in this case, knowing that one individual is contaminated increases the probability that more individuals in the pool are as well, which should make the detection easier.

For the sake of completeness, here we also consider the risk of defective sampling (e.g. that the swabs fails to collect viral load in an infected individual), which we denote $\zeta$. The probability of having a negative pool result given that there is $k$ positive samples within the pool reads, according to the model presented in Eq. \eqref{eqn:groupModel}: 
\begin{align}
\mathbb{P}\left[-\lvert k+ \right] &= \sum_{j=1}^{k} \binom{k}{j} \zeta^{k-j}  (1-\zeta)^{j} \mathbb{P}\left[\log_2\left(\sum_{i=1,\ldots,j} C_i/N\right) > \dcens \right].
\label{eq:falsenegativegeneral}
\end{align}
Under the two assumptions that: 
\begin{enumerate}
    \item the viral load distribution spans several order of magnitudes (e.g. log-normal distributed), so that, following Eq. (\ref{eq:min_viralload}):
\begin{align}
\mathbb{P}\left[\log_2\left(\sum_{i=1,\ldots,j} C_i/N\right) > d_\mathrm{max} \right] = \mathbb{P}\left[\mathrm{min}_{i=1,\ldots,j}( \log_2(C_i)) > d^{(N)}_\mathrm{max} \right],
\end{align}
with $d^{(N)}_\mathrm{max} = \dcens - \log_2(N)$.
    \item  the viral loads (not the infection status) between the k infected individuals are independent, in which case:
\begin{align}
\mathbb{P}\left[\mathrm{min}_{i=1,\ldots,j}( \log_2(C_i)) > d^{(N)}_\mathrm{max} \right] = \mathbb{P}\left[\log_2(C_1) >  d^{(N)}_\mathrm{max}) \right]^j,
\end{align}
\end{enumerate}
we find that Eq. \eqref{eq:falsenegativegeneral} takes the simple expression:
\begin{align}
\mathbb{P}\left[-\lvert k+ \right] &=  \left(\zeta + (1-\zeta)(1-\mathbb{P}\left[\log_2(C_1) <  d^{(N)}_\mathrm{max}) \right]\right)^k,\\
&= \left( 1 - (1-\zeta)\mathbb{P}\left[\log_2(C_1)< d^{(N)}_\mathrm{max})\right]\right)^k.
\end{align}
In Fig. \ref{si:correlated}, in the case of correlated samples, we find that the false negative risk in pooling is greatly reduced if there is more than one positive sample in the pool. The origin of such false-negative reduction is the large variability in viral load and the fact that the amplification technique is particularly sensitive to the highest viral load in the sample. Such false-negative reduction is robust to the presence of a finite risk of defective sampling $\zeta = 5\%$.

\begin{figure}[h]
\centering
\includegraphics[width=8cm]{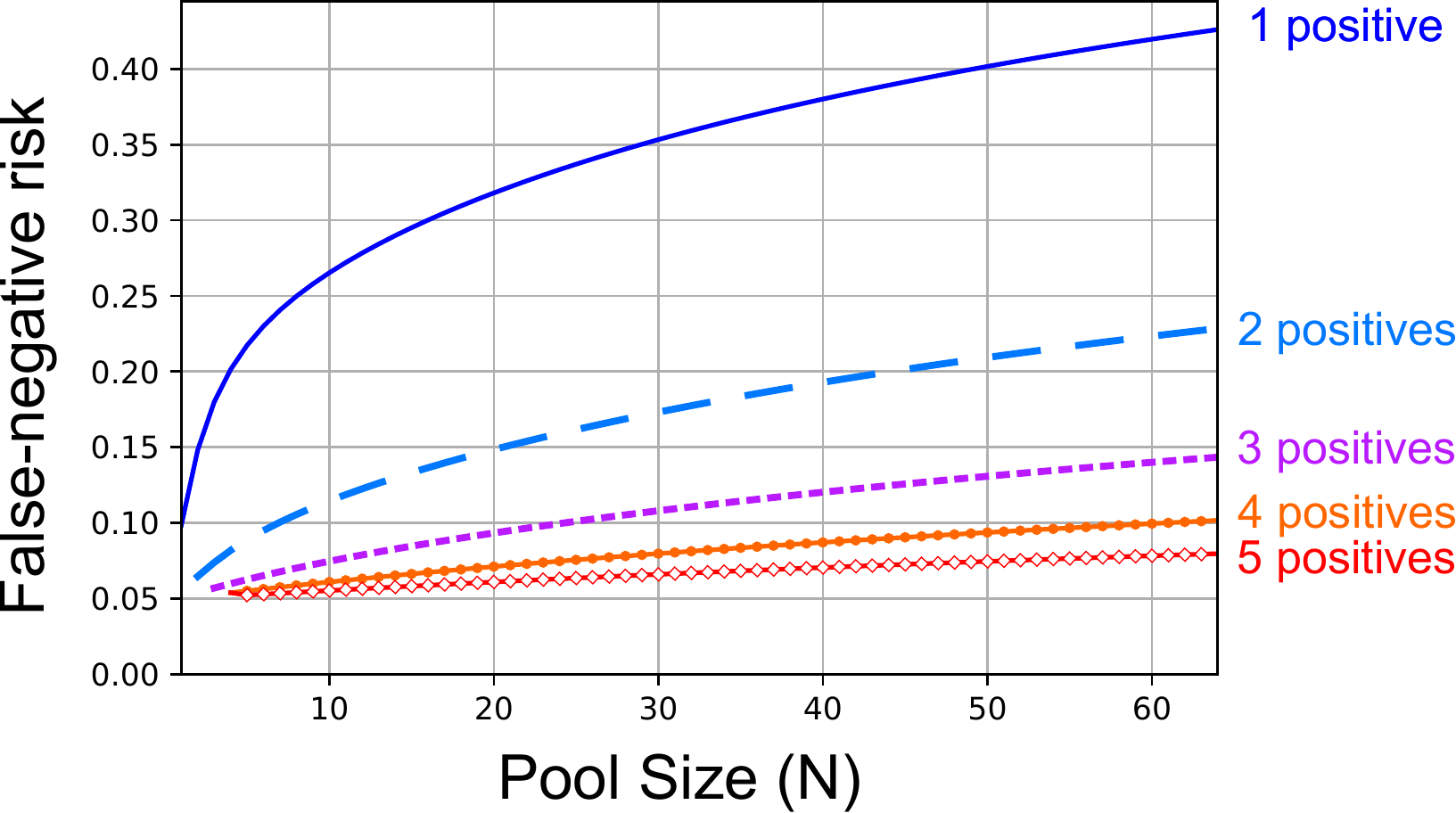}
\caption{Evaluation of the total risk of false negatives estimated according to Eq. (\ref{eq:falsenegativegeneral} as a function of the pool size $N$ for several values of the number of positive samples in the pool $k=1$ (blue solid line); $k=2$ (cyan dashed line); $k=3$ (magenta line); $k=4$ (diamond orange line); $k=5$ (circle red line). We consider a risk that the sample is defective $\zeta = 0.05$. }
\label{si:correlated}
\end{figure}








\end{document}